\documentclass[12pt]{article}
\usepackage[margin=1in]{geometry}
\usepackage{amsmath,amssymb,amsthm,mathtools,bm}
\usepackage{booktabs}
\usepackage{graphicx}
\usepackage{enumitem}
\usepackage{microtype}
\usepackage{natbib}
\usepackage[colorlinks=true,allcolors=blue]{hyperref}
\usepackage{cleveref}

\newtheorem{theorem}{Theorem}[section]
\newtheorem{proposition}[theorem]{Proposition}
\newtheorem{corollary}[theorem]{Corollary}
\newtheorem{example}[theorem]{Example}

\theoremstyle{definition}
\newtheorem{definition}[theorem]{Definition}
\newtheorem{remark}[theorem]{Remark}

\newcommand{\E}{\mathbb E}
\newcommand{\Pp}{\mathbb P}
\newcommand{\R}{\mathbb R}
\newcommand{\Var}{\operatorname{Var}}
\newcommand{\Cov}{\operatorname{Cov}}
\newcommand{\cum}{\operatorname{cum}}
\newcommand{\diag}{\operatorname{diag}}
\newcommand{\LS}{\mathrm{LS}}

\newcommand{\trans}{\mathsf T}
\newcommand{\dd}{\mathrm d}

\title{Linear Regression Beyond Covariance:\\
Latent Slope Spectra from Higher-Order Cumulants$^\ast$}
\author{Anirban Nath and Ming Yuan\\Columbia University}
\date{(\today)}

\begin{document}
\maketitle

\footnotetext[1]{
Address for correspondence: Department of Statistics, Columbia University, 1255 Amsterdam Avenue, New York, NY 10027. Email: ming.yuan@columbia.edu.}

\begin{abstract}
Least squares summarizes the joint distribution of a response and regressor through second-order information. When an observed regressor is composed of latent components carrying different relationships with the response, this summary can average over distinct linear effects. We study the additional linear information contained in higher-order mixed cumulants and formalize it through the \emph{latent slope spectrum}. Ordinary least squares is a variance-weighted barycenter of this spectrum, whereas higher-order cumulants can separate its components when the spectrum is sufficiently parsimonious. In the scalar one-confounder model, fourth-order cumulants identify the unordered pair $\{\beta,\beta+\gamma\}$ through a $2\times2$ linear system and a quadratic equation. We develop estimation and root-$n$ inference under regular separation and show that detection and estimation diverge near degeneracy: slope multiplicity remains detectable on an $n^{-1/4}$ local scale even when the individual slopes are not uniformly regularly estimable. In multiple regression, a quadratic tensor recurrence yields a test of one versus multiple effective coefficient vectors that does not require the vectors themselves to be identified; multiplicity can therefore be consistently detectable even when the recurrence operator is noninjective. The same quadratic diagnostic can retain signal in directions where a first-order higher-cumulant restriction cancels exactly. The results distinguish representation, identification, detection, and estimation of latent linear heterogeneity, while making explicit the limits of structural and causal interpretation from observational data alone.
\end{abstract}

\section{Introduction}

Linear regression is commonly described either as a structural model or as a functional of the joint distribution of a response $Y$ and covariates $X$. Under a structural specification,
\begin{equation*}
Y=X^\trans\beta+\varepsilon,
\qquad \varepsilon\ \text{independent of }X,
\end{equation*}
the coefficient $\beta$ is a parameter of a data-generating mechanism. Under a distribution-free interpretation, by contrast, the population least-squares coefficient is the functional
\begin{equation*}
\beta_{\LS}=\Var(X)^{-1}\Cov(X,Y),
\end{equation*}
where, for simplicity, we assume that all variables are centered. This interpretation is attractive because it is always defined under mild moment assumptions, but it also makes clear what information least squares discards: $\beta_{\LS}$ depends only on the covariance matrix of $(X,Y)$.

This observation motivates three questions that run through the paper.  \emph{What does a single regression coefficient summarize when an observed regressor aggregates components carrying different slopes?}  \emph{What additional linear structure, if any, is present in the joint distribution beyond covariance?}  And, if such structure exists, \emph{which parts of it are representable, identifiable, detectable, and estimable from data?}  These questions are deliberately broader than causal identification: they concern the statistical information in the observed law before one assigns structural labels to it.

This reduction is consequential when the observed regressor contains latent components that have different relationships with the response. To fix ideas, consider the scalar model:
\begin{equation}
X=H+\nu,
\qquad
Y=\beta X+\gamma H+\varepsilon
=(\beta+\gamma)H+\beta\nu+\varepsilon,
\label{eq:one-confounder-intro}
\end{equation}
where $H$, $\nu$, and $\varepsilon$ are mutually independent, or more generally independent to a prescribed finite order as defined below. If second moments exist, then
\begin{equation*}
\beta_{\LS}
=\frac{\Var(H)(\beta+\gamma)+\Var(\nu)\beta}
{\Var(H)+\Var(\nu)}.
\end{equation*}
Thus covariance sees only one variance-weighted average of two distinct slopes $\{\beta, \beta+\gamma\}$. In the structural interpretation, the discrepancy between this barycenter and $\beta$ is the familiar omitted-variable bias. The central observation of this paper is that higher-order mixed cumulants can separate the two latent slopes that covariance averages together.

For an integer $m\ge2$, define
\begin{equation*}
z_k^{(m)}
=\cum(\underbrace{X,\ldots,X}_{m-k},\underbrace{Y,\ldots,Y}_{k}),
\qquad k=0,\ldots,m-1,
\end{equation*}
$m$th-order mixed cumulants that contain at least one regressor argument. Under \eqref{eq:one-confounder-intro}, multilinearity and finite-order independence give
\begin{equation*}
z_k^{(m)}
=\kappa_m(H)(\beta+\gamma)^k+\kappa_m(\nu)\beta^k,
\end{equation*}
where $\kappa_m(\cdot)$ denotes the $m$th-order cumulant. The sequence $k\mapsto z_k^{(m)}$ is therefore a two-exponential sequence. Equivalently, it is the moment sequence of the signed two-atomic measure
\begin{equation*}
\mu_m=\kappa_m(H)\,\delta_{\beta+\gamma}+\kappa_m(\nu)\,\delta_\beta.
\end{equation*}
The support $\{\beta,\beta+\gamma\}$ is the simplest example of what we call the \emph{latent slope spectrum}. It is invisible to covariance alone.

The latent slope spectrum is not proposed as a replacement for the least-squares coefficient.  The two objects answer different questions.  Least squares remains the optimal linear predictor under squared error, while the spectrum describes heterogeneity in linear relationships carried by latent components of the regressor.  Much of the paper is about when the second question is statistically meaningful rather than merely algebraically representable.

In particular, for fourth-order cumulants, writing $z_k=z_k^{(4)}$, the sequence satisfies
\begin{equation*}
z_{k+2}=c_1z_{k+1}+c_0z_k,\qquad k=0,1,
\end{equation*}
where $c_1=2\beta+\gamma$ and $c_0=-\beta(\beta+\gamma)$. Hence
\begin{equation*}
\begin{pmatrix}z_0&z_1\\z_1&z_2\end{pmatrix}
\begin{pmatrix}c_0\\c_1\end{pmatrix}
=\begin{pmatrix}z_2\\z_3\end{pmatrix},
\end{equation*}
and the two latent slopes are the roots of $t^2-c_1t-c_0=0$. Moreover,
\begin{equation*}
z_0z_2-z_1^2=\kappa_4(H)\kappa_4(\nu)\gamma^2.
\end{equation*}
Thus the fourth-order problem becomes degenerate precisely when the two slopes coalesce or when one latent component is invisible to fourth-order cumulants. This identity yields a direct diagnostic for one versus two effective fourth-order slopes and, under the latent omitted-variable model, partially characterizes omitted-variable bias. It does not by itself orient the two structural labels.

The paper develops this perspective in four directions.  \emph{First}, we introduce the latent slope spectrum and show that ordinary least squares is its variance-weighted barycenter.  We also establish a saturation result: at a fixed cumulant order, sufficiently many slope atoms can interpolate an arbitrary cumulant tensor.  Mere existence of a latent-slope representation is therefore weak evidence; the substantive statistical content lies in a parsimonious representation whose spectrum can be learned from the observed law.  \emph{Second}, we establish finite-order identification through Hankel matrices and annihilating polynomials.  Among the order-$m$ mixed cumulants containing at least one regressor argument, up to $\lfloor m/2\rfloor$ distinct effective slopes can be resolved under nondegeneracy.  In the one-confounder model this identifies the unordered pair $\{\beta,\beta+\gamma\}$, but not its structural labels, so signed omitted-variable bias can be partially identified without additional orientation information.  The Hankel and annihilating-polynomial algebra is a form of the classical Prony method for finite exponential sums; see, for example, \citet{potts2010}.  Our contribution is to connect this algebra to random-design regression and to the statistical distinction among representation, identification, detection, and estimation.  \emph{Third}, for the practically important fourth-order two-slope case, we develop an explicit estimator requiring only a $2\times2$ linear system and a quadratic equation, derive regular root-$n$ inference, and characterize a weak-identification regime in which a determinant diagnostic detects $n^{-1/4}$-scale slope separation even though the individual roots are not uniformly regularly estimable.  This route avoids direct low-rank approximation of higher-order tensors, which can be ill posed \citep{desilva2008}.  \emph{Fourth}, we extend the recurrence to multiple regression and show that testing one versus multiple effective coefficient vectors is strictly weaker than identifying the vectors themselves.  A quadratic tensor diagnostic can be calibrated without injectivity of the recurrence operator, has an $n^{-1/4}$ local separation scale, and can remain informative in directions where a first-order higher-cumulant diagnostic cancels exactly.  Under the latent omitted-variable model, absence of omitted-variable bias implies these quadratic restrictions, while the resulting test is consistent against alternatives carrying a visible quadratic fourth-order signal.  As a further extension, we show that cumulants of several orders can share a common annihilating polynomial, which can repair weak or absent information at a particular order.

Our framework is related to several strands of work. The distribution-free interpretation of least squares is emphasized by \citet{buja2019} and has close connections to the econometric perspective of \citet{white1982, white1994, angrist2009}. Methods for hidden confounding often introduce additional observed information, structural restrictions, or low-dimensional latent-factor assumptions; examples include instrumental-variable methods \citep{bowden1990}, spectral deconfounding \citep{cevid2020}, low-complexity latent-cause methods \citep{janzing2011}, and multiple-cause models \citep{wangblei2019}. A different line exploits non-Gaussianity and higher moments for causal identification or goodness-of-fit, including \citet{hoyer2008}, \citet{wangdrton2023}, cumulant-based latent-confounder recovery by \citet{cai2023}, higher-moment goodness-of-fit tests for structural equation models by \citet{schkoda2025}, and higher-order identification in latent-variable LiNGAM by \citet{tramontano2025}.

The closest regression-based comparison is higher-order least squares (HOLS) of \citet{schultheiss2024}. HOLS asks whether a single regression coefficient satisfies additional higher-order residual restrictions. Here we ask what population structure is revealed when such a one-coefficient description fails: mixed cumulants can encode a finite spectrum of latent slopes whose support is identifiable under a parsimonious latent-linear representation. This leads to questions that do not arise in a one-coefficient goodness-of-fit formulation, including how many slopes a cumulant order can resolve, partial identification of structural labels, weak separation of coalescing slopes, and, in multiple regression, testing coefficient multiplicity even when the individual coefficient vectors are not identified. The linear restriction $C(b)$ introduced below is closest in spirit to a first-order higher-moment residual diagnostic, whereas the quadratic restrictions $(M_2,M_3)$ arise from the two-coefficient recurrence and are not redundant with that first-order information: we exhibit a balanced design in which $C(b)$ vanishes identically for every coefficient separation while the quadratic signal remains nonzero.

Accordingly, we regard the latent slope spectrum as an identified population feature only within the stated low-complexity latent-linear model, rather than as a universally defined functional of $P_{XY}$. The saturation results below make this qualification explicit: sufficiently rich algebraic slope representations exist very generally, and substantive statistical content comes from parsimony and identifiability.

The perspective is therefore deliberately intermediate between a purely predictive and a fully structural view of regression.  We do not argue that every statistically distinguishable latent slope should receive a causal interpretation, nor that higher-order information should routinely replace second-order summaries.  Rather, the framework makes visible a choice that is often implicit: whether the target of regression is a single second-order projection or a richer description of linear heterogeneity in the observed law.  The distinction becomes especially consequential when multiplicity is detectable even though its individual components cannot be stably estimated or identified.

The remainder of the paper is organized as follows. \Cref{sec:setup} introduces finite-order independence and the latent slope spectrum. \Cref{sec:identification} gives the Hankel-matrix identification theory and implications for omitted-variable bias. \Cref{sec:fourth} develops fourth-order estimation and inference, including a one-versus-two-slope diagnostic and local alternatives. \Cref{sec:multiorder} presents a brief cross-order extension for settings in which a component is weak or invisible at one cumulant order. \Cref{sec:multivariate} develops the multivariate recurrence, tests of one versus multiple effective coefficient vectors without full identification, and the stronger injective identification result. \Cref{sec:simulation} reports numerical experiments, and \Cref{sec:discussion} concludes. Proofs are collected in the appendices.

\section{Finite-order independence and latent slope spectra}
\label{sec:setup}

Let $W_1,\ldots,W_r$ be random variables with sufficiently many moments. Their joint cumulant is given by 
\[
\cum(W_1,\ldots,W_r)=\sum_\pi (|\pi|-1)!(-1)^{|\pi|-1}\prod_{B\in\pi}{\mathbb E}\left(\prod_{i\in B}W_i\right),
\]
where the sum runs over all partitions $\pi$ of $[r]$. In particular, we write $\kappa_r(W)=\cum(W,\ldots,W)$. For background on cumulants and their tensor representation, see \citet{mccullagh2018tensor}. We use two basic facts: cumulants are multilinear, and a joint cumulant vanishes when its arguments can be partitioned into two independent groups.

\begin{definition}[Finite-order independence]
\label{def:m-independence}
Random elements $U_1,\ldots,U_s$ are \emph{jointly independent through order $m$}, written $U_1\perp_m\cdots\perp_mU_s$, if every mixed joint cumulant of total order at most $m$ that contains arguments from at least two distinct groups vanishes.
\end{definition}

When moments through order $m$ exist, \Cref{def:m-independence} is equivalent to factorization of mixed moments through total degree $m$. More specifically, $U_1\perp_m U_2$ if and only if, for any nonnegative integers $m_1,m_2$ such that $m_1+m_2\le m$ and any pair of polynomials $P_{m_1}$ and $P_{m_2}$ of degrees at most $m_1$ and $m_2$, respectively, $P_{m_1}(U_1)$ and $P_{m_2}(U_2)$ are uncorrelated:
\[
{\mathbb E}[P_{m_1}(U_1)P_{m_2}(U_2)]={\mathbb E}[P_{m_1}(U_1)]{\mathbb E}[P_{m_2}(U_2)].
\]
Such a polynomial-uncorrelated formulation interpolates between uncorrelatedness ($m=2$) and full independence. We opt for the cumulant formulation because it makes the algebra immediate.

Consider a scalar latent linear model:
\begin{equation}
X = \sum_{j=1}^r U_j,
\qquad
Y = \sum_{j=1}^r \theta_j U_j + \varepsilon,
\label{eq:scalar_latent_model}
\end{equation}
where $U_1,\ldots,U_r,\varepsilon$ are jointly independent through order $m$. Note that for a fixed order $m$, only components with a nonzero $m$th cumulant
are visible to the order-$m$ mixed cumulants.

\begin{definition}[Order-$m$ latent slope spectrum]
\label{def:latent_slope_spectrum}
Under \eqref{eq:scalar_latent_model}, the order-$m$ latent slope
spectrum is
\[
\Theta_m
=
\bigl\{
\theta_j:
\kappa_m(U_j)\neq 0
\bigr\},
\]
with repeated values counted only once.  For brevity, we shall omit the subscript $m$ whenever it is clear from the context.
\end{definition}

The one-confounder model
\[
X = H+\nu,
\qquad
Y = \beta X+\gamma H+\varepsilon
= (\beta+\gamma)H+\beta\nu+\varepsilon
\]
is obtained by taking
\[
U_1=H,\qquad \theta_1=\beta+\gamma,
\qquad
U_2=\nu,\qquad \theta_2=\beta.
\]
More generally, with several scalar omitted components,
\begin{equation}
X=\nu+\sum_{j=1}^q H_j,
\qquad
Y=\beta\nu+\sum_{j=1}^q(\beta+\gamma_j)H_j+\varepsilon,
\label{eq:omitted_variables}
\end{equation}
so that $r=q+1$ and the associated slopes are
\begin{equation}
\Theta
=
\{\beta,\beta+\gamma_1,\ldots,\beta+\gamma_q\}.
\label{eq:latent_slope_set}
\end{equation}

For $k=0,\ldots,m-1$, recall that
\begin{equation*}
z_k^{(m)}
=
\operatorname{cum}
\bigl(
\underbrace{X,\ldots,X}_{m-k},
\underbrace{Y,\ldots,Y}_{k}
\bigr).
\end{equation*}

\begin{theorem}[Finite-exponential representation]
\label{thm:finite_exponential}
Suppose \eqref{eq:scalar_latent_model} holds and
$U_1,\ldots,U_r,\varepsilon$ are jointly independent through order
$m$. Then
\begin{equation}
	z_k^{(m)}
	=
	\sum_{j=1}^r
	a_j^{(m)}\theta_j^k,
	\qquad
	a_j^{(m)}=\kappa_m(U_j),
	\qquad
	k=0,\ldots,m-1.
	\label{eq:finite_exponential}
\end{equation}
Equivalently,
$\{z_k^{(m)}\}_{k=0}^{m-1}$ is the truncated moment sequence of the
signed discrete measure
\begin{equation*}
	\mu_m
	=
	\sum_{j=1}^r
	a_j^{(m)}\delta_{\theta_j}.
\end{equation*}
\end{theorem}

Theorem~\ref{thm:finite_exponential} is the basic algebraic reduction
used throughout the paper.  A single-slope regression produces a
geometric sequence: if
\[
Y=\theta X+\varepsilon,
\qquad
X\perp_m\varepsilon,
\]
then
\[
z_k^{(m)}=\kappa_m(X)\theta^k.
\]
Multiple effective slopes therefore appear as a finite mixture of
geometric sequences.

\paragraph{Second order is universally representable.}
At order two, e.g., $m=2$, the latent linear model imposes no restriction beyond
existence of second moments.  This is simply the population
least-squares decomposition.

Suppose $(X,Y)$ has finite second moments and
$\operatorname{Var}(X)>0$.  Define
\[
\theta_{\mathrm{LS}}
=
\frac{\operatorname{Cov}(X,Y)}
{\operatorname{Var}(X)},
\qquad
U_1=X,
\qquad
\varepsilon
=
Y-\theta_{\mathrm{LS}}X.
\]
Then
\[
X=U_1,
\qquad
Y=\theta_{\mathrm{LS}}U_1+\varepsilon,
\]
and
\[
U_1\perp_2\varepsilon.
\]
Consequently every nondegenerate pair $(X,Y)$ with finite second
moments admits the structural representation
\eqref{eq:scalar_latent_model} with $m=2$ and $r=1$.

More generally, least squares collapses a latent spectrum to its variance-weighted barycenter, as the following result indicates.
\begin{proposition}[OLS barycenter]
\label{prop:ols-barycenter}
Suppose \eqref{eq:scalar_latent_model} holds with pairwise uncorrelated $U_1,\ldots,U_r,\varepsilon$, and let $\sigma_j^2=\Var(U_j)>0$. Then
\begin{equation}
	\beta_{\LS}=\sum_{j=1}^r\omega_j\theta_j,
	\quad {\rm where} \quad
	\omega_j=\frac{\sigma_j^2}{\sum_{\ell=1}^r\sigma_\ell^2}.
	\label{eq:ols-general-barycenter}
\end{equation}
In particular,
\begin{equation}
	\min_j\theta_j\le\beta_{\LS}\le\max_j\theta_j.
	\label{eq:convex-hull}
\end{equation}
\end{proposition}

In the two-slope case, once $\theta_-=\min\{\theta_1,\theta_2\}$ and $\theta_+=\max\{\theta_1,\theta_2\}$ are known, the second-order variance weights are identified from $\beta_{\LS}$:
\begin{equation*}
\omega_+=\frac{\beta_{\LS}-\theta_-}{\theta_+-\theta_-},
\qquad \omega_-=1-\omega_+.
\end{equation*}
Multiplying by $\Var(X)$ recovers the two latent variance contributions indexed by their slopes.

While least squares may be viewed as the unique one-slope
second-order representation obtained by requiring the residual to be
second-order independent of the regressor, the situation changes at
higher orders: a low-dimensional representation becomes restrictive,
but a sufficiently rich representation is again able to interpolate
arbitrary cumulant information.

\paragraph{Saturation at a fixed cumulant order.}
Let
\[
K_m(X,Y)
=
\operatorname{cum}_m\bigl((X,Y)^\top\bigr)
\]
denote the symmetric order-$m$ joint cumulant tensor, and write
\[
v(\theta)
=
\begin{pmatrix}
1\\ \theta
\end{pmatrix},
\qquad
e_2
=
\begin{pmatrix}
0\\ 1
\end{pmatrix}.
\]

The next result shows that, if enough slope components are allowed, a
single cumulant order places no restriction on the order-$m$ cumulant
tensor.

\begin{proposition}[Order-$m$ cumulant saturation]
\label{prop:order_m_saturation}
Suppose $(X,Y)$ has finite $m$th moments.  Fix any $m$ distinct real
numbers
\[
\vartheta_1,\ldots,\vartheta_m.
\]
Then there exist unique real numbers
$a_1^{(m)},\ldots,a_m^{(m)}$ such that
\begin{equation*}
z_k^{(m)}
=
\sum_{j=1}^m
a_j^{(m)}\vartheta_j^k,
\qquad
k=0,\ldots,m-1.
\end{equation*}
Furthermore, the complete order-$m$ cumulant tensor satisfies
\begin{equation*}
K_m(X,Y)
=
\sum_{j=1}^m
a_j^{(m)}
v(\vartheta_j)^{\otimes m}
+
b_m e_2^{\otimes m},
\end{equation*}
where
\begin{equation*}
b_m
=
\kappa_m(Y)
-
\sum_{j=1}^m
a_j^{(m)}\vartheta_j^m.
\end{equation*}
\end{proposition}

\Cref{prop:order_m_saturation} indicates that $m$ slope components, together with a response-only component, always suffice to represent an arbitrary order-$m$ cumulant tensor. Moreover, the $m$ slopes in this representation may be chosen arbitrarily, provided they are distinct. This is useful for interpreting
the identification results below. In particular, it suggests that existence of a latent-slope representation at one cumulant order is by itself weak information: once $m$ slope atoms are permitted, every order-$m$ cumulant tensor can be represented, and the slopes themselves can even be fixed arbitrarily in advance.  The substantive content comes from a \emph{low-complexity} representation in which the number of effective slopes is small enough to be identified from the observed cumulants.

\paragraph{A common slope dictionary across several orders.}
A similar saturation statement holds simultaneously for all cumulant
orders up to $m$.

\begin{corollary}[Common-slope algebraic representation]
\label{cor:common_slope_dictionary}
Suppose $(X,Y)$ has finite moments through order $m$, and fix any $m$
distinct real numbers
$\vartheta_1,\ldots,\vartheta_m$.
For every $\ell=2,\ldots,m$, there exist weights
$a_1^{(\ell)},\ldots,a_m^{(\ell)}$ and $b_\ell$ such that
\begin{equation*}
	K_\ell(X,Y)
	=
	\sum_{j=1}^m
	a_j^{(\ell)}
	v(\vartheta_j)^{\otimes \ell}
	+
	b_\ell e_2^{\otimes \ell}.
\end{equation*}
Equivalently,
\begin{equation}
	z_k^{(\ell)}
	=
	\sum_{j=1}^m
	a_j^{(\ell)}
	\vartheta_j^k,
	\qquad
	k=0,\ldots,\ell-1.
	\label{eq:common_dictionary_mixed}
\end{equation}
Thus the same set of $m$ slopes can algebraically represent all
cumulant tensors through order $m$, although the weights are allowed
to depend on the cumulant order.
\end{corollary}

The results in \Cref{prop:order_m_saturation,cor:common_slope_dictionary} are statements about cumulant
tensors.  They do \emph{not} imply that every joint distribution of
$(X,Y)$ admits a structural representation
\eqref{eq:scalar_latent_model} through order $m$.
For such a structural representation to exist, the order-specific
weights must be compatible across $\ell$.  In particular, one would
need random variables $U_1,\ldots,U_r,\varepsilon$ such that
\[
a_j^{(\ell)}
=
\kappa_\ell(U_j),
\qquad
b_\ell
=
\kappa_\ell(\varepsilon),
\qquad
\ell=2,\ldots,m,
\]
for the same latent variables at every order.  Thus, for each $j$,
the sequence
\[
\bigl(
a_j^{(2)},a_j^{(3)},\ldots,a_j^{(m)}
\bigr)
\]
must be a realizable truncated cumulant sequence of a probability
distribution, and similarly for the response-only component.
These compatibility conditions are automatic at order two, but they are not
automatic at higher orders.

Accordingly, throughout the paper we distinguish a
\emph{structural latent-linear model}, in which common latent random
variables generate cumulants across orders, from an
\emph{algebraic cumulant representation} at one or several fixed
orders.

\paragraph{Identification versus saturation.}
The distinction above leads to a useful hierarchy.  At a fixed order
$m$, the mixed-cumulant vector
\[
\bigl(
z_0^{(m)},\ldots,z_{m-1}^{(m)}
\bigr)
\]
contains $m$ scalar coordinates.  An $r$-slope representation
\[
z_k^{(m)}
=
\sum_{j=1}^r
a_j^{(m)}\theta_j^k
\]
contains $r$ slope parameters and $r$ order-$m$ weights.  As shown in
\Cref{sec:identification}, when
\(
2r\le m
\)
and the usual nondegeneracy conditions hold, these $2r$ parameters, i.e., $\{a_j^{(m)},\theta_j: 1\le j\le r\}$,
are identifiable through a Hankel system.  By contrast,
\Cref{prop:order_m_saturation} shows that allowing
$r=m$ makes the order-$m$ representation saturated: every mixed
cumulant vector can be interpolated, even with a preassigned set of
distinct slopes.

Thus the informative regime is not the mere existence of a
higher-order latent-linear representation, but the existence of a
\emph{parsimonious} one.  Higher-order cumulants reveal latent
structure precisely when the effective slope spectrum is sufficiently
small relative to the cumulant order to be identifiable.

\section{Identification from a single cumulant order}
\label{sec:identification}

\subsection{Identification by Hankel matrices}
Assume that the order-$m$ spectrum contains $r$ distinct slopes $\theta_1,\ldots,\theta_r$ and corresponding nonzero weights $a_1,\ldots,a_r$. Suppress the superscript $(m)$ for $z_k$s and $a_j$s. 
For $s\ge1$, define the Hankel matrix
\begin{equation*}
\mathcal H_s(z)=\bigl(z_{i+j}\bigr)_{i,j=0}^{s-1}.
\end{equation*}

\begin{theorem}[Identification of the latent slope spectrum]
\label{thm:prony-identification}
Suppose $m\ge2r$, the slopes $\theta_1,\ldots,\theta_r$ are distinct, and $a_j\neq0$ for every $j$. Then $\mathcal H_r(z)$ is nonsingular and the unordered slope spectrum $\{\theta_1,\ldots,\theta_r\}$ is uniquely determined by $z_0,\ldots,z_{2r-1}$.
\end{theorem}

The latent slope spectrum can be given more explicitly under the conditions of \Cref{thm:prony-identification}. Let $c=(c_0,\ldots,c_{r-1})^\trans$ solve
\begin{equation}
\mathcal H_r(z)c=
\begin{pmatrix}
	z_r\\z_{r+1}\\\vdots\\z_{2r-1}
\end{pmatrix}.
\label{eq:annihilator-system}
\end{equation}
Then the slopes are exactly the roots of
\begin{equation}
p(t)=t^r-\sum_{\ell=0}^{r-1}c_\ell t^\ell.
\label{eq:annihilator-poly}
\end{equation}
Once the roots are known, the weights $a_j$ are uniquely determined by the first $r$ equations in \eqref{eq:finite_exponential}.

\begin{corollary}[Identification capacity of one cumulant order]
\label{cor:capacity}
The $r$ distinct effective slopes of the scalar latent linear model \eqref{eq:scalar_latent_model} are generically identifiable from $z_k^{(m)}$, $k=0,\ldots,m-1$ if
\[
r\le r_{\max}=\left\lfloor\frac m2\right\rfloor.
\]
In particular, the latent slope set \eqref{eq:latent_slope_set} of the scalar omitted-variable model \eqref{eq:omitted_variables} is generically identifiable from these cumulants if
\begin{equation*}
q\le \left\lfloor\frac m2\right\rfloor-1.
\end{equation*}
\end{corollary}

This is an identification statement, not an existence statement. With more components, an order-$m$ cumulant may still admit many latent linear representations, but the nodes cannot generically be recovered from the $m$ available mixed-cumulant coordinates alone.

\subsection{Partial identification of omitted-variable bias}

It is instructive to consider the one-confounder case \eqref{eq:one-confounder-intro} and investigate the implications of these identification results. If
\[
X=H+\nu,\qquad Y=(\beta+\gamma)H+\beta\nu+\varepsilon,
\]
then exchanging $H$ and $\nu$ yields the same observed $(X,Y)$ while interchanging the roles of $\beta$ and $\beta+\gamma$. This observational symmetry is fundamental.

Consider the fourth-order case $m=4$. We can write
\begin{equation}
z_k:=z_k^{(4)}=a\theta_1^k+b\theta_2^k,\qquad k=0,1,2,3,
\label{eq:two-node}
\end{equation}
where
\[
a=\kappa_4(H),\quad b=\kappa_4(\nu),\quad
\{\theta_1,\theta_2\}=\{\beta,\beta+\gamma\}.
\]
Define
\begin{equation*}
H(z)=\begin{pmatrix}z_0&z_1\\z_1&z_2\end{pmatrix},
\qquad
q(z)=\begin{pmatrix}z_2\\z_3\end{pmatrix}.
\end{equation*}

\begin{corollary}[Fourth-order two-slope formula]
\label{cor:fourth-order}
Under \eqref{eq:two-node},
\begin{equation}
	\det H(z)=ab(\theta_1-\theta_2)^2.
	\label{eq:det-h}
\end{equation}
If this determinant is nonzero, define
\begin{equation*}
	\begin{pmatrix}c_0\\c_1\end{pmatrix}=H(z)^{-1}q(z).
\end{equation*}
Then
\begin{equation}
	c_1=\theta_1+\theta_2,\qquad c_0=-\theta_1\theta_2,
	\label{eq:symmetric-roots}
\end{equation}
and hence the two slopes are
\begin{equation}
	\theta_\pm=\frac{c_1\pm\sqrt{c_1^2+4c_0}}{2}.
	\label{eq:theta-pm}
\end{equation}
\end{corollary}

Under \eqref{eq:one-confounder-intro},
\begin{equation*}
\det H(z)=\kappa_4(H)\kappa_4(\nu)\gamma^2.
\end{equation*}
The determinant isolates the three fundamental fourth-order degeneracies: no slope separation, zero fourth cumulant of the confounding component, or zero fourth cumulant of the idiosyncratic component.

\begin{proposition}[Partial identification of omitted-variable bias]
\label{prop:ovb-set}
Suppose the assumptions of \Cref{cor:fourth-order} hold for \eqref{eq:one-confounder-intro}, and let $\theta_-<\theta_+$ denote the identified latent slopes. Then
\[
\{\beta,\beta+\gamma\}=\{\theta_-,\theta_+\},\qquad {\rm and} \qquad |\gamma|=\theta_+-\theta_-.
\]
Without an additional labeling restriction,
\begin{equation*}
	\beta\in\{\theta_-,\theta_+\},
\end{equation*}
and the omitted-variable bias belongs to
\begin{equation*}
	\mathcal B=\{\beta_{\LS}-\theta_-,\;\beta_{\LS}-\theta_+\}.
\end{equation*}
\end{proposition}

A scientifically justified side condition can orient the pair. Examples include external knowledge about which component has larger variance, known support restrictions, or auxiliary measurements tied to one latent component. Such assumptions should be stated separately from the information supplied by cumulants themselves.

\section{Estimation and inference}
\label{sec:fourth}

We now discuss how the identification results from the previous section translate into statistical estimation and inference. Because sampling variability increases rapidly with cumulant order, we focus primarily on fourth-order methods, which already capture the main identification and weak-separation phenomena. \Cref{sec:multiorder} then briefly considers how information from several cumulant orders can be combined when a single order is weak or uninformative.

\subsection{Sample cumulants}
Let $(X_i,Y_i)_{i=1}^n$ be i.i.d. copies of $(X,Y)$. For $k=0,1,2,3$, define the empirical fourth mixed cumulant by applying the cumulant functional to the empirical distribution,
\begin{equation*}
\widehat z_k
=\cum_{\Pp_n}(\underbrace{X,\ldots,X}_{4-k},\underbrace{Y,\ldots,Y}_{k}).
\end{equation*}
One may equivalently use unbiased $k$-statistics; the difference is $O_{\Pp}(n^{-1})$ and does not affect first-order theory.

Let $M_4$ collect all raw moments $\E[X^aY^b]$ with $a+b\le4$, and $\widehat M_4$ be the empirical moment vector. Under the assumption that
\begin{equation}
\E[(|X|+|Y|)^8]<\infty,
\label{eq:eighth-moment}
\end{equation}
we have
\[
\sqrt{n}[\text{vec}(\widehat{M}_4)-\text{vec}(M_4)]\overset{d}{\longrightarrow}N(0,\Sigma_M),
\]
where $\Sigma_M$ is the covariance matrix of the monomial vector defining $\widehat M_4$. Write $z=\phi(M_4)$ for the polynomial map from raw moments to $(z_0,z_1,z_2,z_3)^\trans$, then the usual sandwich formula yields
\begin{equation}
\sqrt n(\widehat z-z)\overset{d}{\longrightarrow}N(0,\Omega),
\label{eq:z-clt}
\end{equation}
where
\begin{equation}
\Omega=D\phi(M_4)\Sigma_M D\phi(M_4)^\trans.
\label{eq:omega}
\end{equation}
A consistent estimator $\widehat\Omega$ can be obtained by plug-in. The representation \eqref{eq:omega} is convenient computationally: it involves moments only through total order eight and can be generated automatically from the monomial basis.

\subsection{Estimator of the latent slopes}
Define
\[
\widehat H=\begin{pmatrix}\widehat z_0&\widehat z_1\\\widehat z_1&\widehat z_2\end{pmatrix},
\qquad
\widehat q=\begin{pmatrix}\widehat z_2\\\widehat z_3\end{pmatrix}.
\]
Whenever $\widehat H$ is nonsingular, set
\begin{equation*}
\widehat c=\begin{pmatrix}\widehat c_0\\\widehat c_1\end{pmatrix}
=\widehat H^{-1}\widehat q,
\end{equation*}
and, when the discriminant is nonnegative,
\begin{equation*}
\widehat\theta_\pm
=\frac{\widehat c_1\pm\sqrt{\widehat c_1^2+4\widehat c_0}}{2},
\qquad \widehat\theta_-\le\widehat\theta_+.
\end{equation*}
Under the regular two-slope model, nonsingularity and a positive discriminant hold with probability tending to one.

Put $A=H(z)^{-1}$. The derivative of $c=H(z)^{-1}q(z)$ is
\begin{equation}
D_zc=A
\begin{pmatrix}
-c_0&-c_1&1&0\\
0&-c_0&-c_1&1
\end{pmatrix}.
\label{eq:jac-c}
\end{equation}
If $\Delta=c_1^2+4c_0>0$, then
\begin{equation}
D_c\theta_\pm=
\left(\pm\Delta^{-1/2},\;\frac12\pm\frac{c_1}{2\sqrt\Delta}\right).
\label{eq:jac-root}
\end{equation}

\begin{theorem}[Asymptotic normality of the fourth-order slope estimator]
\label{thm:theta-clt}
Assume \eqref{eq:eighth-moment}, $\det H(z)\neq0$, and $\theta_-<\theta_+$. Then with probability tending to one, $\widehat\theta_-$ and $\widehat\theta_+$ are distinct reals. Furthermore,
\begin{equation*}
\sqrt n
\begin{pmatrix}
	\widehat\theta_- -\theta_-\\
	\widehat\theta_+ -\theta_+
\end{pmatrix}
\overset{d}{\longrightarrow}N(0,V_\theta),
\end{equation*}
where
\begin{equation*}
V_\theta=G\Omega G^\trans,
\qquad
G=\begin{pmatrix}D_c\theta_-\\D_c\theta_+\end{pmatrix}D_zc.
\end{equation*}
The plug-in estimator $\widehat V_\theta$ is consistent.
\end{theorem}

The same argument yields joint inference for the least-squares coefficient and the two candidate omitted-variable biases. Because $(\widehat\beta_{\LS},\widehat z)$ is a smooth function of empirical moments through degree four, the vector
\begin{equation*}
\left(\widehat\beta_{\LS},\widehat\theta_-,\widehat\theta_+,
\widehat\beta_{\LS}-\widehat\theta_-,
\widehat\beta_{\LS}-\widehat\theta_+\right)
\end{equation*}
is asymptotically normal under the same regularity conditions.

\subsection{A diagnostic for one versus two effective slopes}
Define
\begin{equation*}
h(z)=z_0z_2-z_1^2.
\end{equation*}
Under a single effective fourth-order slope, $z_k=a\theta^k$ and $h(z)=0$. Under a regular two-slope model, $h(z)=ab(\theta_1-\theta_2)^2$. This suggests
\begin{equation*}
\widehat h=\widehat z_0\widehat z_2-\widehat z_1^2,
\end{equation*}
as a test statistic for one versus two effective fourth-order slopes. Under the structural one-confounder model, rejection is evidence of fourth-order slope heterogeneity associated with omitted-variable effects, subject to the cumulant nondegeneracy conditions below.

\begin{proposition}[Asymptotic test of one effective fourth-order slope]
\label{prop:h-test}
Suppose \eqref{eq:eighth-moment} holds and, under the null that there is a single effective fourth-order slope,
\begin{equation*}
z_k=a\theta^k,\qquad a\neq0.
\end{equation*}
Then
\begin{equation*}
\sqrt n\,\widehat h\overset{d}{\longrightarrow}N(0,\tau_h^2),
\end{equation*}
where
\begin{equation}
\tau_h^2=\nabla h(z)^\trans\Omega\nabla h(z),
\qquad
\nabla h(z)=(z_2,-2z_1,z_0,0)^\trans.
\label{eq:tau-h}
\end{equation}
If, in addition, $\tau_h^2>0$ and $\widehat\tau_h$ is a consistent plug-in estimator of $\tau_h$, then
\[
Z_h=\frac{\sqrt n\,\widehat h}{\widehat\tau_h}\overset{d}{\longrightarrow}N(0,1).
\]
so a two-sided normal critical value yields an asymptotic level-$\alpha$ test of the null that the fourth-order mixed-cumulant sequence has one effective slope.
\end{proposition}

The null is deliberately phrased as a statement about the fourth-order cumulant sequence, not as ``no confounding.'' A confounded component with zero fourth cumulant is invisible to this test.

\begin{remark}[Closed form under Gaussian noise and a null degeneracy]
\label{rem:tau-closed}
It is instructive to consider the case under Gaussian noise where a closed-form expression can be obtained. Suppose $\gamma=0$, so $Y=\beta X+\varepsilon$ with $\varepsilon\sim N(0,\sigma_\varepsilon^2)$ independent of $X$, $\Var(X)=1$, and $X$ symmetric. The influence function of $\widehat h$ then factors as
\[
\kappa_4(X)\,(X^2-\E X^2)\,(\varepsilon^2-\E\varepsilon^2),
\]
and \eqref{eq:tau-h} becomes
\begin{equation}
\tau_h^2
=2\sigma_\varepsilon^4\,\kappa_4(X)^2\bigl(\E[X^4]-1\bigr).
\label{eq:tau-gaussian}
\end{equation}
Thus $\tau_h=0$ whenever $\kappa_4(X)=0$. For $X=\sqrt\rho\,H_0+\sqrt{1-\rho}\,\nu_0$ with symmetric independent factors,
\[
\kappa_4(X)=\rho^2\kappa_4(H_0)+(1-\rho)^2\kappa_4(\nu_0),
\]
so the marginal cumulant can vanish by cancellation even though $\kappa_4(H_0)$ and $\kappa_4(\nu_0)$ are both nonzero. Alternative identification of two slopes requires the latent cumulants to be nonzero. First-order regularity of the determinant test requires the further condition $\kappa_4(X)\neq0$.
\end{remark}

\subsection{Weak separation and an \texorpdfstring{$n^{-1/4}$}{n to the -1/4} local scale}
Equation \eqref{eq:det-h} shows that the determinant signal is quadratic in the slope separation.
\begin{proposition}[Local separation for the determinant diagnostic]
\label{prop:local-n14}
Consider a triangular array of two-slope models \eqref{eq:two-node} with weights $a_n\to a\neq0$ and $b_n\to b\neq0$ and slope gap
\begin{equation}
\delta_n=\theta_{2,n}-\theta_{1,n}=tn^{-1/4}
\label{eq:delta-n14}
\end{equation}
for fixed $t\in\R$. Suppose the joint cumulant CLT holds uniformly along the sequence and $\tau_h^2\to\tau_0^2\in(0,\infty)$. Then
\begin{equation}
\sqrt n\,\widehat h\overset{d}{\longrightarrow}N(abt^2,\tau_0^2).
\label{eq:local-h}
\end{equation}
Thus the determinant test has nontrivial local power when the slopes are separated on the $n^{-1/4}$ scale.
\end{proposition}

This concerns detection of slope multiplicity, not regular estimation of individual roots. As $\delta_n\to0$, $H(z)$ becomes singular and the root map ceases to be differentiable, so $(\widehat\theta_-,\widehat\theta_+)$ is not estimable at the root-$n$ rate uniformly down to a repeated root. The $n^{-1/4}$ scale in \eqref{eq:delta-n14} is the scale on which $\sqrt n\,h(z_n)$ stays of constant order. It is not an estimation rate for the individual slopes. At $t=0$ the two-root parameterization is not identified: a geometric sequence satisfies many quadratic recurrences.

The condition number of the sample matrix $\widehat H$ is a numerical diagnostic, not a test of this degeneracy. When $H(z)$ is singular, sampling error moves the zero singular value off zero, and a typical $\widehat H$ need not look ill-conditioned. The studentized statistic $|Z_h|=\sqrt n\,|\widehat h|/\widehat\tau_h$ tracks the size of the determinant relative to \eqref{eq:tau-h}. It is used below only as a descriptive measure of identification strength. A cutoff such as $|Z_h|\ge2$ is not a basis for post-selection Wald inference.

\subsection{A cross-order extension}
\label{sec:multiorder}
A component may be invisible at one order but visible at another. This motivates combining several orders while exploiting the fact that the latent slopes do not change with cumulant order.

For $\ell\in\mathcal M\subset\{3,4,5,\ldots\}$, define
\begin{equation*}
z_k^{(\ell)}=\sum_{j=1}^ra_j^{(\ell)}\theta_j^k,
\qquad k=0,\ldots,\ell-1,
\end{equation*}
where $a_j^{(\ell)}=\kappa_\ell(U_j)$ but the nodes $\theta_j$ are common. Let
\[
p_c(t)=t^r-\sum_{h=0}^{r-1}c_ht^h.
\]
Every order obeys the same recurrence
\begin{equation*}
z_{k+r}^{(\ell)}-\sum_{h=0}^{r-1}c_hz_{k+h}^{(\ell)}=0,
\qquad k=0,\ldots,\ell-r-1.
\end{equation*}
Stack the sample recurrence residuals into $\widehat R(c)$ and define
\begin{equation}
\widehat c=\arg\min_c\widehat R(c)^\trans W_n\widehat R(c),
\label{eq:gmm-c}
\end{equation}
where $W_n$ is positive definite. The estimated slopes are the roots of $p_{\widehat c}$.

This construction has three advantages. First, one order can identify a component that is absent at another. Second, when several orders are informative, the recurrence equations are overidentified and can improve precision. Third, the overidentifying restrictions provide a specification diagnostic: different cumulant orders must imply the same slope polynomial.

\begin{proposition}[Multi-order minimum distance]
\label{prop:multiorder-gmm}
Suppose the stacked population recurrence $R(c)=0$ uniquely identifies $c=c_0$, the relevant sample cumulants satisfy a joint root-$n$ CLT, and $W_n\to W$ with $W$ positive definite. If $D_cR(c_0)$ has full column rank, then the minimum-distance estimator in \eqref{eq:gmm-c} is root-$n$ consistent and asymptotically normal. If the asymptotic covariance of the stacked recurrence residuals is nonsingular, weighting by its inverse is efficient within this minimum-distance class. If the roots of $p_{c_0}$ are simple, the estimated latent slopes are also root-$n$ asymptotically normal.
\end{proposition}

\section{Multiple Regression: Detection Without Full Identification}
\label{sec:multivariate}

The distinction between detection and estimation becomes even more
pronounced in multiple regression.  In the scalar model, a determinant
of the fourth-order Hankel matrix can detect the presence of two
effective slopes in regimes where the individual roots are only weakly
identified.  In the multivariate setting, an analogous phenomenon is
possible at a more fundamental level: one can test whether the data are
compatible with a single effective coefficient vector even when the
individual coefficient vectors under a two-component model are not
separately identifiable.

\subsection{Two latent coefficient vectors and a tensor recurrence}
\label{sec:multiv_recurrence}

Let $X\in\mathbb{R}^p$ and consider
\begin{equation}
X = U+V,
\qquad
Y = \eta^\top U+\beta^\top V+\varepsilon,
\label{eq:multiv_two_component}
\end{equation}
where $U,V\in\mathbb{R}^p$ and $U$, $V$, and $\varepsilon$ are
jointly independent through order $m$.  No independence among the
coordinates within $U$ or within $V$ is required.

The omitted-variable model
\begin{equation}
X=\Delta H+\nu,
\qquad
Y=X^\top\beta+H^\top\gamma+\varepsilon,
\label{eq:multiv_ovb_model}
\end{equation}
is a special case.  Assume $\Delta^\top\Delta=I_q$ without loss of generality. \eqref{eq:multiv_ovb_model} becomes
\eqref{eq:multiv_two_component} when we take $U=\Delta H$, $V=\nu$, and $\eta=\beta+\Delta\gamma$.

For $k=0,\ldots,m-1$, define the symmetric mixed-cumulant tensor
\begin{equation*}
Z_k
=
\operatorname{cum}
\bigl(
\underbrace{X,\ldots,X}_{m-k},
\underbrace{Y,\ldots,Y}_{k}
\bigr)
\in
\operatorname{Sym}^{m-k}(\mathbb{R}^p).
\end{equation*}
For a symmetric tensor $T$, let $T\times v$ denote contraction of one
mode with a vector $v$, and let $T:Q$ denote contraction of two modes
with a symmetric matrix $Q$. More generally, we write $T[v_1,\ldots, v_k]$ as the contraction of $T$ along the first $k$ modes by vector $v_1,\ldots,v_k$.

Write
\[
K_U=\operatorname{cum}_m(U),
\qquad
K_V=\operatorname{cum}_m(V).
\]
Multilinearity and finite-order independence give
\begin{equation}
Z_k
=
K_U[\underbrace{\eta,\ldots,\eta}_{k}]
+
K_V[\underbrace{\beta,\ldots,\beta}_{k}].
\label{eq:multiv_Z_representation}
\end{equation}

\begin{theorem}[Multivariate cumulant recurrence]
\label{thm:multiv_recurrence}
Suppose \eqref{eq:multiv_two_component} holds and $U$, $V$, and
$\varepsilon$ are jointly independent through order $m$.  Let
\begin{equation*}
	s=\beta+\eta,
	\qquad
	Q=\frac{1}{2}
	\bigl(
	\beta\eta^\top+\eta\beta^\top
	\bigr).
\end{equation*}
Then, for $k=0,\ldots,m-3$,
\begin{equation}
	Z_{k+2}
	=
	Z_{k+1}\times s
	-
	Z_k:Q.
	\label{eq:multiv_recurrence}
\end{equation}
\end{theorem}

The important feature of
\Cref{thm:multiv_recurrence} is that the recurrence is linear in
$(s,Q)$.  Neither the recurrence itself nor the testing procedures
developed below require $(s,Q)$ to be identified. To fix ideas, for the remainder of this section we take $m=4$, so that
\begin{align}
Z_2
&=
Z_1\times s-Z_0:Q,
\label{eq:multiv_rec_1}\\
Z_3
&=
Z_2\times s-Z_1:Q.
\label{eq:multiv_rec_2}
\end{align}

\subsection{One effective coefficient vector}
\label{sec:multiv_one_vector}

Let
\[
\Sigma_U=\operatorname{Var}(U),
\qquad
\Sigma_V=\operatorname{Var}(V),
\qquad
\Sigma_X=\Sigma_U+\Sigma_V,
\]
and assume $\Sigma_X$ is nonsingular.  Define
\begin{equation*}
b
=
\beta_{\mathrm{LS}}
=
\Sigma_X^{-1}\operatorname{Cov}(X,Y),
\qquad
\delta=\eta-\beta.
\end{equation*}
Since
\[
\operatorname{Cov}(X,Y)
=
\Sigma_U\eta+\Sigma_V\beta,
\]
we have
\begin{equation}
b
=
\beta+A\delta,
\qquad
A=\Sigma_X^{-1}\Sigma_U.
\label{eq:multiv_barycenter}
\end{equation}
Consequently,
\begin{equation}
\eta-b=(I-A)\delta,
\qquad
\beta-b=-A\delta.
\label{eq:multiv_centered_slopes}
\end{equation}

Thus $b$ is the multivariate analogue of the scalar least-squares
barycenter.  The null hypothesis of a single effective coefficient
vector is
\begin{equation}
H_0:\quad \eta=\beta,
\label{eq:multiv_H0}
\end{equation}
under which $b=\eta=\beta$. Under \eqref{eq:multiv_H0}, the second-order recurrence in
\eqref{eq:multiv_recurrence} collapses to the first-order recurrence
\begin{equation*}
Z_{k+1}=Z_k\times b,
\qquad
k=0,1,2.
\end{equation*}
This already gives observable restrictions that do not involve
identification of $(s,Q)$. More specifically, let
\begin{equation*}
C(b)
=
Z_1-Z_0\times b
\in
\operatorname{Sym}^3(\mathbb{R}^p).
\end{equation*}
Then under $H_0$, $C(b)=0$. This is closest in spirit to the higher-order least-squares diagnostic of \citet{schultheiss2024}.

Another set of restriction follows directly from the two-coefficient recurrence given by Theorem \ref{thm:multiv_recurrence}.
Under \eqref{eq:multiv_H0}, $s=2b$ and $Q=bb^\top$,
or equivalently $ss^\top-4Q=0$.
This immediately implies that, under the one effective coefficient hypothesis $H_0$, $M_2(b)=0$ and $M_3(b)=0$ where
\begin{align*}
M_2(b)
&=
Z_2
-2Z_1\times b
+Z_0:(bb^\top)\in \operatorname{Sym}^2(\mathbb{R}^p),
\\
M_3(b)
&=
Z_3
-2Z_2\times b
+Z_1:(bb^\top)\in \mathbb{R}^p.
\end{align*}

\subsection{Linear and quadratic diagnostics}
\label{sec:multiv_diagnostics}

More generally, under the two coefficient vector model \eqref{eq:multiv_two_component}, $C(b)$, $M_2(b)$ and $M_3(b)$ are population signals of the separation between the two coefficient vectors, e.g., $\delta=\eta-\beta$.

\begin{proposition}[Observable restrictions under one effective coefficient vector]
\label{prop:multiv_testing_restrictions}
Suppose the assumptions of
\Cref{thm:multiv_recurrence} hold with $m=4$, and suppose
$\Sigma_X$ is nonsingular. Then
\begin{equation}
	C(b)
	=
	K_U\times\{(I-A)\delta\}
	-
	K_V\times(A\delta),
	\label{eq:multiv_C_signal}
\end{equation}
\begin{equation}
	M_2(b)
	=
	K_U[(I-A)\delta,(I-A)\delta]
	+
	K_V[A\delta,A\delta],
	\label{eq:multiv_M2_signal}
\end{equation}
and
\begin{equation}
	M_3(b)
	=
	K_U[\eta,(I-A)\delta,(I-A)\delta]
	+
	K_V[\beta,A\delta,A\delta],
	\label{eq:multiv_M3_signal}
\end{equation}
where $A$ and $\delta$ are given by \eqref{eq:multiv_barycenter}.
%
\end{proposition}

\Cref{prop:multiv_testing_restrictions} separates two types of
information.  The diagnostic $C(b)$ has a signal that is linear
in $\delta$, whereas $M_2(b)$ and $M_3(b)$ have signals that are
quadratic in $\delta$.  The latter are the direct multivariate
analogues of the scalar determinant
\[
z_0z_2-z_1^2,
\]
whose signal is quadratic in the separation between two scalar
slopes.

As in the scalar case, the restrictions are statements about the
fourth-order cumulant structure.  A component that is invisible to
fourth-order cumulants, or an exact cancellation between the two
terms in
\eqref{eq:multiv_C_signal}--\eqref{eq:multiv_M3_signal}, may make a
particular diagnostic insensitive to an alternative.  Thus failure
to reject should not, without further assumptions, be interpreted as
proof that no latent component is present.

\subsection{A multivariate test that does not require identification}
\label{sec:multiv_test}

Let
\[
\widehat b
=
\widehat\Sigma_X^{-1}
\widehat{\operatorname{Cov}}(X,Y)
\]
be the ordinary least-squares estimator.  Replace $Z_k$ by the
empirical fourth-order mixed-cumulant tensors $\widehat Z_k$ and
define
\begin{align*}
\widehat M_2
&=
\widehat Z_2
-2\widehat Z_1\times\widehat b
+\widehat Z_0:
(\widehat b\widehat b^\top),\\
\widehat M_3
&=
\widehat Z_3
-2\widehat Z_2\times\widehat b
+\widehat Z_1:
(\widehat b\widehat b^\top).
\end{align*}

Let $\operatorname{vec}(T)$ collect the distinct coordinates of a
symmetric tensor $T$, and put
\begin{equation*}
\psi_Q(P)
=
\begin{pmatrix}
	\operatorname{vec}\{M_2(b)\}\\
	M_3(b)
\end{pmatrix},
\qquad
\widehat\psi_Q
=
\begin{pmatrix}
	\operatorname{vec}(\widehat M_2)\\
	\widehat M_3
\end{pmatrix}.
\end{equation*}
The dimension of $\psi_Q$ is
\[
d_Q
=
\frac{p(p+1)}{2}+p
=
\frac{p(p+3)}{2}.
\]

\begin{theorem}[Test of one effective fourth-order coefficient vector]
\label{thm:multiv_wald_test}
Suppose \eqref{eq:multiv_two_component} holds,
\[
\mathbb{E}\bigl(\|X\|+|Y|\bigr)^8<\infty,
\]
and $\Sigma_X$ is nonsingular.  
\begin{itemize}
	\item[(a)] Under
$H_0:\eta=\beta$,
\begin{equation}
	\sqrt n\,\widehat\psi_Q
	\ \xrightarrow{d}\
	N(0,\Omega_Q),
	\label{eq:multiv_psi_clt}
\end{equation}
where $\Omega_Q$ is obtained by applying the delta method to the
empirical raw moments through total degree four.	
	\item[(b)] 	Let $\widehat\Omega_Q$ be a consistent estimator of $\Omega_Q$.
If $\operatorname{rank}(\Omega_Q)=r_Q$ is locally constant, then
\begin{equation}
	T_{Q,n}
	=
	n
	\widehat\psi_Q^\top
	\widehat\Omega_Q^\dagger
	\widehat\psi_Q
	\ \xrightarrow{d}\
	\chi^2_{r_Q},
	\label{eq:multiv_Q_stat}
\end{equation}
under $H_0$, where $\dagger$ denotes the Moore--Penrose inverse.
\item[(c)] Under a fixed alternative $P$, let $\Omega_Q(P)$ denote the probability limit of
$\widehat\Omega_Q$ and suppose its rank is locally constant. If
\begin{equation}
\psi_Q(P)^\top\Omega_Q(P)^\dagger\psi_Q(P)>0,
\label{eq:multiv_visible_alt}
\end{equation}
then
\[
T_{Q,n}\xrightarrow{p}\infty.
\]
\end{itemize}
\end{theorem}

It is worth pointing out that neither the null distribution nor consistency against such visible fixed
alternatives in Theorem \ref{thm:multiv_wald_test} requires the recurrence operator
$\mathcal L_Z$ to be injective. In fact, under the assumptions of \Cref{thm:multiv_wald_test}, if
$\mathcal L_Z$ is noninjective and the fixed alternative satisfies the
visibility condition \eqref{eq:multiv_visible_alt}, then the unordered pair
$\{\beta,\eta\}$ need not be identified by the fourth-order recurrence,
whereas
\[
T_{Q,n}\xrightarrow{p}\infty.
\]
Thus multiplicity of effective coefficient vectors can be consistently
detectable even when the coefficient vectors themselves are not identified.

The covariance calculation in
\Cref{thm:multiv_wald_test} is the direct multivariate analogue of
the scalar raw-moment delta-method calculation in \Cref{sec:fourth}: collect the
monomials in $(X^\top,Y)^\top$ through total degree four, estimate
their covariance, and propagate it through the smooth map defining
$\widehat\psi_Q$.

One may analogously construct a first-order statistic from
\[
\psi_L(P)
=
\operatorname{vec}\{C(b)\},
\]
or an omnibus statistic obtained by stacking
$\psi_L$ and $\psi_Q$.  We return to the distinction between these
statistics below.

\subsection{Local alternatives and two detection scales}
\label{sec:multiv_local}

The quadratic form of
\eqref{eq:multiv_M2_signal}--\eqref{eq:multiv_M3_signal} produces the
same $n^{-1/4}$ local scale that appears in the scalar determinant
diagnostic.

Consider a triangular array such that
\begin{equation}
\delta_n
=
\eta_n-\beta_n
=
n^{-1/4}h,
\label{eq:multiv_local_delta}
\end{equation}
where $h\in\mathbb{R}^p$ is fixed.  Suppose
\[
\beta_n\to\theta,
\qquad
\eta_n\to\theta,
\qquad
A_n\to A,
\]
and the fourth-order cumulant tensors converge to $K_U$ and $K_V$.
Define
\[
u=(I-A)h,
\qquad
v=Ah.
\]
Then
\begin{align}
\sqrt n\,M_{2,n}
&\longrightarrow
K_U[u,u]+K_V[v,v],
\notag\\
\sqrt n\,M_{3,n}
&\longrightarrow
K_U[\theta,u,u]+K_V[\theta,v,v].
\label{eq:multiv_local_M3}
\end{align}
The expression for $M_{3,n}$ also contains a term cubic in $\delta_n$ through
$\eta_n=\beta_n+\delta_n$.  Under \eqref{eq:multiv_local_delta} that term is
$O(n^{-3/4})$, so it vanishes after multiplication by $\sqrt n$ and does not
alter the local limit above.

Let
\begin{equation*}
\mu_Q(h)
=
\begin{pmatrix}
	\operatorname{vec}
	\left\{
	K_U[u,u]+K_V[v,v]
	\right\}\\[3pt]
	K_U[\theta,u,u]+K_V[\theta,v,v]
\end{pmatrix}.
\end{equation*}

\begin{proposition}[$n^{-1/4}$ local separation]
\label{prop:multiv_local_power}
Suppose the assumptions of
\Cref{thm:multiv_wald_test} hold uniformly along the triangular array
\eqref{eq:multiv_local_delta}, and suppose
$\Omega_{Q,n}\to\Omega_{Q,0}$ with locally constant rank $r_Q$.
Then
\begin{equation}
	\sqrt n\,\widehat\psi_Q
	\ \xrightarrow{d}\
	N\bigl(
	\mu_Q(h),
	\Omega_{Q,0}
	\bigr).
	\label{eq:multiv_local_normal}
\end{equation}
If $\mu_Q(h)$ belongs to the column space of $\Omega_{Q,0}$, then
\begin{equation*}
	T_{Q,n}
	\ \xrightarrow{d}\
	\chi^2_{r_Q}
	\bigl(
	\lambda_Q(h)
	\bigr),
\end{equation*}
where
\begin{equation*}
	\lambda_Q(h)
	=
	\mu_Q(h)^\top
	\Omega_{Q,0}^\dagger
	\mu_Q(h).
\end{equation*}
Hence the quadratic recurrence diagnostic has nontrivial local power
against $n^{-1/4}$ separation whenever
$\lambda_Q(h)>0$.
\end{proposition}

This result concerns detection of coefficient-vector multiplicity,
not estimation of $\beta$ and $\eta$ separately.  In particular,
the $n^{-1/4}$ scale does not imply an $n^{-1/4}$ estimation rate for
either coefficient vector.

The linear diagnostic $C(b)$ behaves differently.  Define the linear
map
\begin{equation*}
\mathcal B h
=
K_U\times\{(I-A)h\}
-
K_V\times(Ah).
\end{equation*}
Under
\[
\delta_n=h\,n^{-1/2},
\]
we have
\[
\sqrt n\,C(b_n)
\longrightarrow
\mathcal B h.
\]
Thus the first-order diagnostic has a nonzero $n^{-1/2}$ local signal whenever $\mathcal B h\neq0$.  There may, however, exist directions
$h\neq0$ for which
\begin{equation*}
\mathcal B h=0.
\end{equation*}
Along such first-order cancellation directions, the linear
diagnostic has no local signal.  If at the same time
\begin{equation*}
K_U[(I-A)h,(I-A)h]
+
K_V[Ah,Ah]
\neq0,
\end{equation*}
the quadratic recurrence diagnostic remains informative on the
larger $n^{-1/4}$ scale.

The linear and quadratic diagnostics therefore capture complementary
features of the multivariate latent-slope structure.  An omnibus
procedure may stack both sets of restrictions, while the quadratic
statistic is especially useful for making the connection with the
two-coefficient recurrence and with the scalar weak-separation
phenomenon.

\begin{example}[First-order cancellation with a nonzero quadratic signal]
\label{ex:multiv_cancellation}
Suppose the coordinates of $U$ and $V$ are independent, have common
standardized fourth cumulant $\kappa\neq0$, and
\[
\Sigma_U=\rho I_p,
\qquad
\Sigma_V=(1-\rho)I_p.
\]
Then $A=\rho I_p$ and
\[
K_U=\kappa\rho^2\sum_{j=1}^p e_j^{\otimes4},
\qquad
K_V=\kappa(1-\rho)^2\sum_{j=1}^p e_j^{\otimes4}.
\]
Consequently,
\begin{equation*}
C(b)
=
\kappa\rho(1-\rho)(2\rho-1)
\sum_{j=1}^p \delta_j e_j^{\otimes3}.
\end{equation*}
At $\rho=1/2$, the linear diagnostic therefore vanishes for every
$\delta$, while
\begin{equation*}
M_2(b)
=
\frac{\kappa}{8}
\sum_{j=1}^p \delta_j^2 e_j^{\otimes2},
\end{equation*}
which is nonzero for every $\delta\neq0$.  Thus the quadratic recurrence
restriction can detect coefficient heterogeneity in directions for which the
first-order fourth-cumulant diagnostic is identically blind.
\end{example}

\subsection{Connection with omitted-variable bias}
\label{sec:multiv_ovb}

The testing problem has a direct omitted-variable interpretation
under \eqref{eq:multiv_ovb_model}.  Let
\[
\Sigma_H=\operatorname{Var}(H)
\]
and suppose
\[
\Delta^\top\Delta=I_q,
\qquad
\Sigma_H\succ0,
\qquad
\Sigma_X\succ0.
\]

\begin{corollary}[Structural equivalence of no confounding and no omitted-variable bias]
\label{cor:multiv_ovb}
Under \eqref{eq:multiv_ovb_model},
testing $H_0:\eta=\beta$ is equivalent to testing absence of omitted-variable bias, e.g., $\gamma=0$.
\end{corollary}

The structural qualification in
\Cref{cor:multiv_ovb} is important.  The corollary identifies the structural null
$\eta=\beta$ with absence of omitted-variable bias; it does not assert that the
quadratic restrictions characterize that null without further fourth-order
nondegeneracy.  Outside the latent linear model, failure of the fourth-order
restrictions may also be produced by nonlinearity, heteroskedasticity, or other
misspecification. Conversely, a latent component whose relevant fourth-order
cumulant vanishes, or whose contribution cancels another component, may be
invisible to a test based on $T_{Q, n}$, hereafter referred to as TEST Q.  We therefore interpret the unconditional statistical null
as one effective fourth-order coefficient vector and invoke the omitted-variable
interpretation only under the structural assumptions above.  Within that model,
TEST Q is consistent against alternatives satisfying
\eqref{eq:multiv_visible_alt}.

\subsection{Full identification is a stronger problem}
\label{sec:multiv_identification}

Testing whether $\eta=\beta$ does not require recovery of either
coefficient vector separately.  Full identification is a stronger
question.

At fourth order, define
\begin{equation*}
\mathcal L_Z(s,Q)
=
\left(
Z_1\times s-Z_0:Q,\,
Z_2\times s-Z_1:Q
\right).
\end{equation*}
The recurrence equations require
\[
\mathcal L_Z(s,Q)=(Z_2,Z_3).
\]
\begin{theorem}[Full identification under injectivity]
\label{thm:multiv_identification}
Suppose the assumptions of \Cref{thm:multiv_recurrence} hold with $m=4$.
If $\mathcal L_Z$ is injective, then $(s,Q)$ is uniquely identified from
$(Z_0,Z_1,Z_2,Z_3)$.  Moreover,
\begin{equation*}
ss^\top-4Q
=
(\beta-\eta)(\beta-\eta)^\top.
\end{equation*}
If $\beta\neq\eta$, let $d$ satisfy
\begin{equation}
dd^\top
=
ss^\top-4Q.
\label{eq:multiv_d_factor}
\end{equation}
Then
\begin{equation}
\{\beta,\eta\}
=
\left\{
\frac{s+d}{2},
\frac{s-d}{2}
\right\}.
\label{eq:multiv_beta_eta_recovery}
\end{equation}
\end{theorem}

Thus injectivity permits recovery of the unordered pair of coefficient
vectors.  It is not, however, required for the testing results in
\Cref{thm:multiv_wald_test,prop:multiv_local_power}.  When
$\mathcal L_Z$ is noninjective, different pairs $(s,Q)$ may satisfy
the recurrence, so an arbitrary minimum-norm solution need not
recover the true $\{\beta,\eta\}$.  Nevertheless, the one-vector null
still imposes the directly observable restrictions $C(b)=0$, $M_2(b)=0$ and $M_3(b)=0$.

This distinction mirrors, and strengthens, the scalar separation
between detectability and estimation.  In the scalar problem,
multiplicity can be detectable when the individual roots are weakly
identified.  In multiple regression, multiplicity can remain
detectable even when the individual coefficient vectors are not
identified at all.

\section{Numerical experiments}
\label{sec:simulation}
The simulations are intended to illustrate the statistical boundaries of the framework rather than to stage a broad estimator comparison.  They are organized around four questions that mirror the main theory: whether the scalar spectrum can be estimated under regular identification, what happens as two slopes coalesce, what happens when fourth-order information itself becomes weak, and whether multivariate slope multiplicity can be detected without identifying the individual coefficient vectors.  Unless noted, each design uses $2{,}000$ Monte Carlo replications; dedicated null-calibration runs use $5{,}000$. For scalar root estimators we report median matched error together with RMSE and the real-root probability, because a small number of near-singular samples can dominate mean squared error.

\paragraph{Regular scalar identification.}
Data are generated from
\[
X=H+\nu,
\qquad
Y=\beta X+\gamma H+\varepsilon,
\qquad
H=\sqrt\rho\,H_0,
\qquad
\nu=\sqrt{1-\rho}\,\nu_0,
\]
where $\varepsilon\sim N(0,\sigma_\varepsilon^2)$ is independent of $(H_0,\nu_0)$, $H_0$ and $\nu_0$ have unit variance, and $\beta=1$.  The regular grid uses $\gamma\in\{0.25,0.5,1\}$, $\rho\in\{0.1,0.3,0.5,0.7,0.9\}$, $\sigma_\varepsilon\in\{0.5,1,2\}$, $n\in\{250,500,1000,2500,5000\}$, and Laplace/Laplace, Laplace/Rademacher, and Rademacher/Rademacher factor pairs.

\Cref{fig:barycenter} illustrates the basic phenomenon for Rademacher factors with $\gamma=1$, $\rho=0.5$, and $\sigma_\varepsilon=1$: least squares concentrates around the barycenter $\beta_{\LS}=1.5$, whereas the fourth-order estimator separates the two latent slopes $1$ and $2$.  At $n=1000$, all $2{,}000$ replications return two real roots, the median matched error is $0.184$, the matched-pair RMSE is $0.263$, the delta-method standard error for $\theta_-$ is $1.002$ times its Monte Carlo standard deviation, and the nominal $95\%$ Wald interval covers in $94.4\%$ of replications.  At $n=5000$ the median error falls to $0.079$.

\begin{figure}[ht]
\centering
\includegraphics[width=0.9\textwidth]{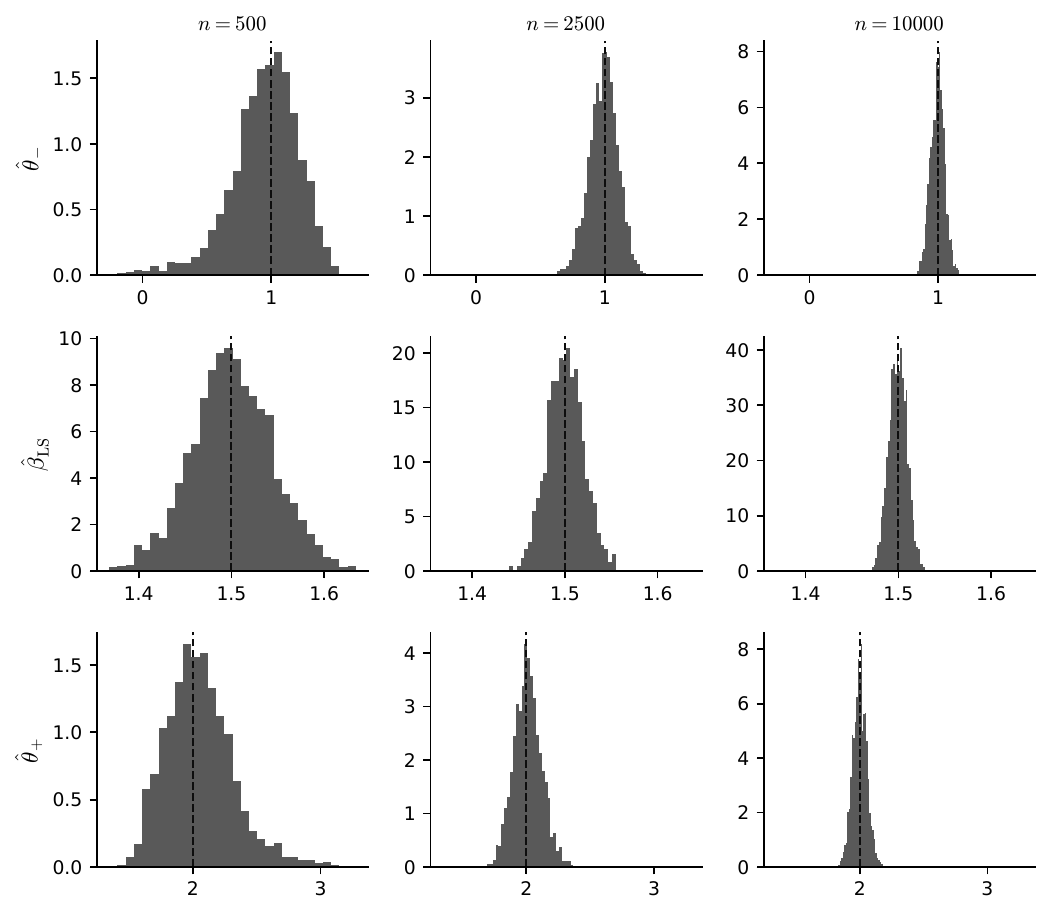}
\caption{Sampling distributions of $\widehat\theta_-$, $\widehat\beta_{\LS}$, and $\widehat\theta_+$ for Rademacher factors, $\gamma=1$, $\rho=0.5$, and $\sigma_\varepsilon=1$.  Dashed lines mark $\theta_-=1$, $\beta_{\LS}=1.5$, and $\theta_+=2$.}
\label{fig:barycenter}
\end{figure}

The same pointwise asymptotics can emerge much more slowly for heavier-tailed factors or smaller separation.  With Laplace factors and $\gamma=1$, the median error is $0.525$ at $n=500$ even though the RMSE is $41.8$, reflecting a few extreme roots; by $n=10000$ the median error is $0.089$, the RMSE is $0.113$, and Wald coverage is $93.9\%$.  With fixed $\gamma=0.5$, convergence is slower still: at $n=50000$, coverage is $92.0\%$ for Rademacher factors and $93.8\%$ for Laplace factors.  Thus \Cref{thm:theta-clt} is pointwise rather than uniform as the two slopes approach one another.

\paragraph{Weak separation.}
Set $\gamma_n=hn^{-1/4}$ with $h\in\{0,0.5,1,2,4,6,8\}$, $n\in\{500,1000,2500,5000,10000\}$, $\rho=0.5$, and $\sigma_\varepsilon=1$.  \Cref{fig:local-power} shows the nominal $5\%$ determinant test.  For Rademacher factors, rejection at $h=0$ stays between $0.036$ and $0.052$ and at $h=4$ between $0.646$ and $0.780$ across the full range of $n$; the curves are nearly invariant to sample size, as predicted by \Cref{prop:local-n14}.  Laplace factors enter this local regime more slowly.  At the same time the individual roots remain weakly identified: at $h=0$ and $n=10000$, the median matched root error is $1.83$ for Rademacher factors and $1.09$ for Laplace factors and does not decrease toward zero.  Detection of multiplicity and estimation of the individual roots are therefore distinct problems.

\begin{figure}[htbp]
\centering
\includegraphics[width=0.8\textwidth]{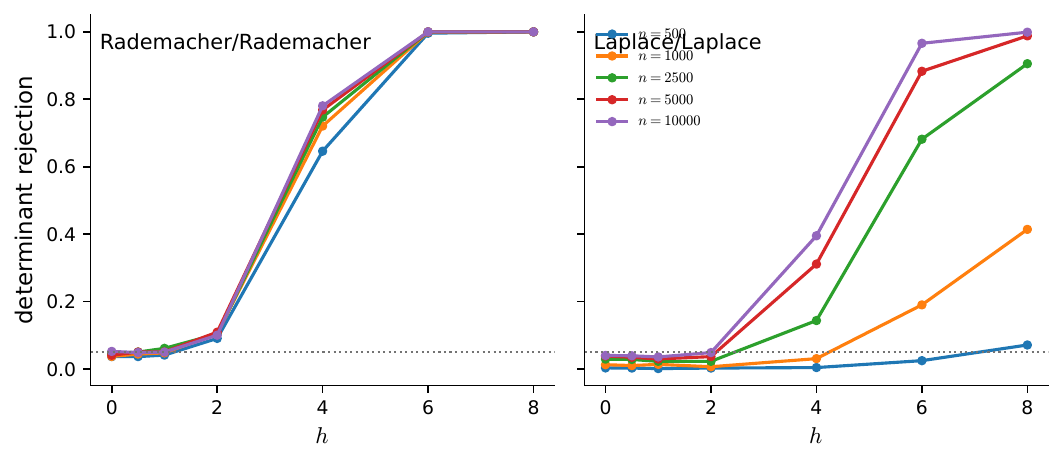}
\caption{Determinant rejection against $h$ in $\gamma_n=hn^{-1/4}$, at $\rho=0.5$ and $\sigma_\varepsilon=1$.  The dotted line is the nominal $5\%$ level.}
\label{fig:local-power}
\end{figure}

\paragraph{Weak fourth-order information.}
To vary fourth-order visibility directly, let
\[
U_t=\sqrt{1-t}\,L+\sqrt t\,R,
\qquad
\kappa_4(U_t)=3(1-t)^2-2t^2,
\]
where $L$ is standardized Laplace, $R$ is Rademacher, and $L\perp R$.  The fourth cumulant vanishes at $t_0=3-\sqrt6$.  In the design with one factor equal to $U_t$, $\gamma=0.5$, $\rho=0.5$, $\sigma_\varepsilon=1$, and $n=5000$, the population condition number of the Hankel matrix exceeds $10^{16}$ at $t_0$, but the median sample condition number is only $103$.  The real-root probability falls to $86.0\%$ and the median matched error rises to $1.01$.  Thus loss of fourth-order identification is visible statistically even when the realized sample Hankel matrix does not appear numerically singular.  This also reinforces that a numerical condition number is not a substitute for an identification diagnostic.

\paragraph{Detection without identification in multiple regression.}
We next use the multiple-regression setting of \Cref{sec:multivariate}.  The baseline design has $p=3$, $\beta=(1,-0.5,0.5)^\top$, Gaussian noise, and non-diagonal within-block mixing.  The $q=1$ Rademacher design is deliberately noninjective: the population recurrence operator $\mathcal L_Z$ has rank $4$ for the $9$ coordinates of $(s,Q)$, so the minimum-norm estimator of $\{\beta,\eta\}$ is not consistent.  Nevertheless, the quadratic statistic $T_{Q,n}$ is well calibrated: under $H_0$, its covariance rank is $9$ in every Rademacher null cell, its nominal $5\%$ rejection probability is $0.049$ at both $n=5000$ and $n=10000$, and at $n=10000$ its Monte Carlo mean and variance are $9.03$ and $17.76$, close to the $\chi^2_9$ values $9$ and $18$.

\Cref{fig:multiv_detect_without_id} gives the central comparison.  At $\|\eta-\beta\|=0.5$, TEST Q rejection increases from $0.216$ at $n=500$ to $1.000$ at $n=10000$, while the median matched error of the minimum-norm coefficient-pair estimator remains essentially unchanged, from $4.19$ to $4.16$.  At separation $1.5$, the rejection probability is already $1.000$ at $n=500$, whereas the matched estimation error remains about $4.95$.  Thus coefficient-vector multiplicity can be consistently detected even though the vectors themselves are not identified.

\begin{figure}[htbp]
\centering
\includegraphics[width=0.85\textwidth]{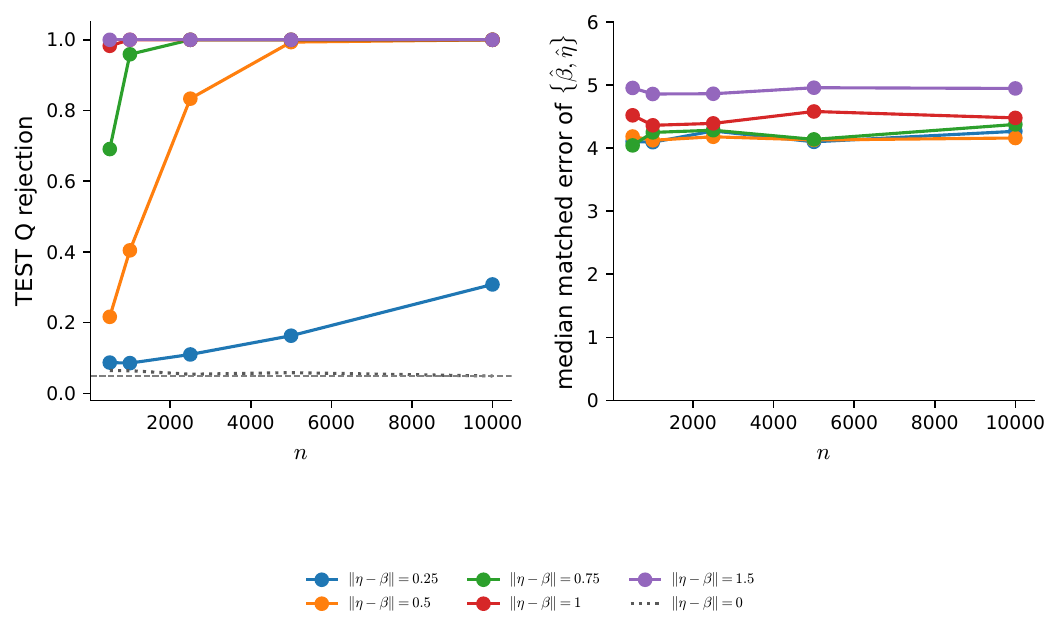}
\caption{Detection without identification in the noninjective $p=3$, $q=1$ Rademacher design.  Left: rejection probability of the quadratic recurrence test at nominal level $5\%$.  Right: median matched error of the minimum-norm estimator of $\{\beta,\eta\}$.  The test becomes powerful while the estimation error does not decrease with $n$.}
\label{fig:multiv_detect_without_id}
\end{figure}

The multivariate local experiment also supports the quadratic scale in \Cref{prop:multiv_local_power}.  Along $\eta_n-\beta=hn^{-1/4}d$, Rademacher rejection curves are nearly invariant to $n$: at $h=2$ rejection lies between $0.136$ and $0.172$, at $h=4$ between $0.856$ and $0.989$, and it is essentially one for $h\ge6$.  Laplace factors again approach the local limit more slowly.

Finally, \Cref{fig:multiv_cancellation} shows that the quadratic recurrence test is not merely a less efficient version of the first-order diagnostic.  In the balanced diagonal design of \Cref{ex:multiv_cancellation}, $C(b)=0$ for every separation while $M_2(b)\neq0$ whenever $\eta\neq\beta$.  At fixed separation $\|\eta-\beta\|=1$, the test based on the linear restriction $C(b)$, referred to as TEST L, rejects in only $2.1\%$ of samples at $n=1000$ and $1.6\%$ at $n=10000$, whereas TEST Q rejects in $94.3\%$ and $100\%$, respectively.

\begin{figure}[htbp]
\centering
\includegraphics[width=0.8\textwidth]{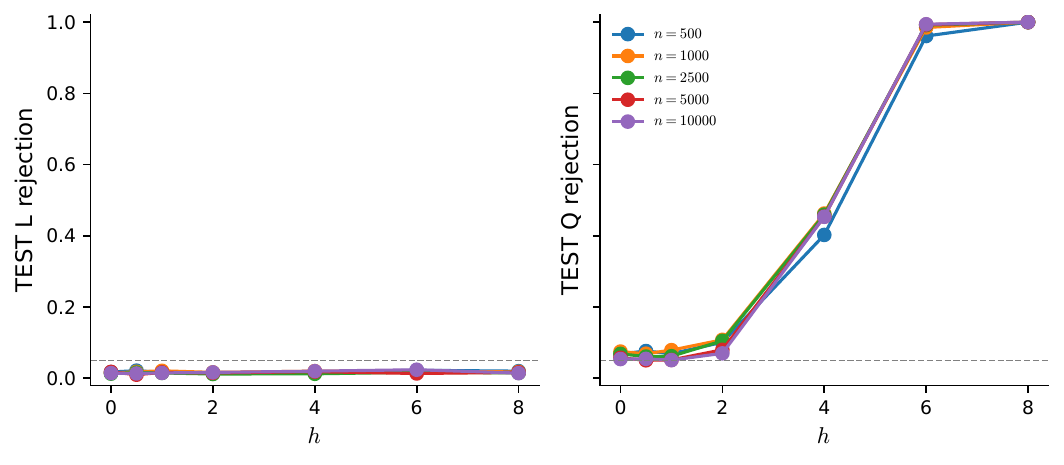}
\caption{First-order cancellation in the balanced diagonal Rademacher design.  The linear diagnostic TEST L is population-blind, while the quadratic recurrence diagnostic TEST Q remains informative.  The dotted line is the nominal $5\%$ level.}
\label{fig:multiv_cancellation}
\end{figure}

Taken together, the simulations support the distinction that drives the theory: higher-order slope multiplicity can become detectable before the individual slopes are stably estimable, and in multiple regression it can remain detectable even when the individual coefficient vectors are not identified.

\section{Discussion}
\label{sec:discussion}
Least squares answers a second-order question: which linear predictor is optimal under squared error, or equivalently which slope is encoded by the covariance of $(X,Y)$. That question remains important, but it is not the only linear information contained in the joint distribution. In latent linear systems, higher-order mixed cumulants form finite exponential sequences whose nodes are coefficient values attached to latent components of the regressor. The resulting latent slope spectrum separates structure that covariance necessarily averages together.

The scalar fourth-order case makes this distinction especially transparent. Four fourth-order mixed cumulants, a $2\times2$ linear system, and one quadratic equation are enough to recover two effective slopes under regular identification. The construction is explicit and avoids direct low-rank tensor approximation, but its limitations are equally transparent. A component with zero fourth cumulant is invisible at fourth order, nearly coincident slopes create weak identification, and the observed law does not by itself label the two endpoints as ``direct'' and ``confounded.'' These are properties of the information available in the distribution rather than defects of a particular algorithm.

The theory and experiments separate three statistical questions. \emph{Identification} asks whether the population cumulant sequence determines a latent spectrum. \emph{Detectability} asks whether the data contain enough signal to reject a one-slope representation. \emph{Estimation} asks whether the individual slopes can be recovered stably. These questions have different answers near degeneracy. Under $\gamma_n=hn^{-1/4}$, the determinant remains on a nontrivial local scale, so multiplicity can be detectable even though the individual roots are not uniformly regularly estimable. Likewise, a cumulant order can be algebraically blind to a component even when another order is informative. This is why the multi-order formulation is useful primarily as a way to repair weak or absent information at one order rather than as an automatic precision improvement.

The numerical results also show the finite-sample price of using higher-order information. For two well-separated slopes, ordinary root-$n$ inference emerges cleanly for bounded latent factors and more slowly for heavier-tailed ones; root errors can have very long right tails, so medians and real-root probabilities are as important as RMSE. Near a vanishing fourth cumulant, statistical identification deteriorates even when the realized sample Hankel matrix does not look numerically singular.

The multivariate results sharpen the separation between detectability and identification.  The quadratic recurrence restrictions can test a one-vector representation without injectivity of $\mathcal L_Z$.  In the noninjective $p=3$ experiment, the rejection probability of TEST Q approaches one under visible fixed alternatives even though the minimum-norm estimator of $\{\beta,\eta\}$ does not improve with $n$.  The balanced cancellation design is stronger still: the first-order fourth-cumulant diagnostic vanishes identically, while the quadratic recurrence remains informative.  This also distinguishes the recurrence test from first-order higher-moment residual diagnostics rather than merely providing a less efficient version of them.  Thus higher-order slope multiplicity can be detectable not only before stable estimation, as in the scalar weak-identification problem, but even when the individual multivariate coefficient vectors are not identified at all.

Two qualifications deserve emphasis. First, multiple effective roots or coefficient vectors in cumulant restrictions do not by themselves establish a particular latent causal mechanism; nonlinear or other forms of misspecification can also create departures from a one-slope representation, while fourth-order cancellation can hide some latent alternatives. Second, injectivity remains essential for recovery of the coefficient pair in \Cref{thm:multiv_identification}.  The testing results do not remove this identification limitation; rather, they show that scientifically meaningful restrictions survive it.  Characterizing primitive conditions for multivariate injectivity therefore remains a separate problem.

These results also clarify the role of the proposed spectrum relative to ordinary regression. If prediction under squared loss is the objective, the least-squares coefficient has an unambiguous role. If the scientific question concerns heterogeneous linear relationships carried by latent components of an observed regressor, however, a single covariance-based coefficient may deliberately average over structure that is itself of interest. The latent slope spectrum provides one population summary of that heterogeneity under the low-complexity model considered here; it is complementary to, rather than a replacement for, the least-squares projection.

A second implication is that nonidentification of the individual latent parameters need not exhaust the inferential content of the observed law. Both the scalar and multivariate results exhibit intermediate regimes in which a one-slope description can be rejected even though the individual slopes are not regularly estimable, or are not identified at all. Thus tests of slope multiplicity can remain informative in settings where parameter recovery is impossible, provided the relevant higher-order signal is visible.

Interpretation nevertheless requires care. The saturation results show that rich algebraic representations are easy to construct, and nonlinear or heteroskedastic misspecification can also violate one-slope restrictions. We therefore regard the cumulant restrictions first as observational statements about the joint law. Structural or causal labels require assumptions or external information that are justified independently of the algebra. This distinction is central to the partial-identification results in the scalar model and to the interpretation of the multivariate tests.

Finally, the gain in structural resolution comes with a statistical cost. Higher-order cumulants can recover features that second moments erase, but heavier tails slow convergence, weak cumulants create near-degeneracy, and individual roots can have long-tailed sampling distributions. The usefulness of information beyond covariance therefore depends jointly on the scientific target and on the strength and stability of the available higher-order signal.

\newpage
\appendix

\section{Proofs for the scalar population results}

\begin{proof}[Proof of \Cref{thm:finite_exponential}]
Expand each occurrence of $X$ and $Y$ in the joint cumulant using multilinearity. Every resulting term is a joint cumulant of $m$ arguments drawn from $U_1,\ldots,U_r,\varepsilon$. By finite-order independence, any term involving at least two distinct groups vanishes. Because $k\le m-1$, at least one argument in $z_k^{(m)}$ is an $X$, so a term involving only $\varepsilon$ is impossible. The only surviving terms are those in which all $m$ arguments arise from a common $U_j$. Such a term contributes $\kappa_m(U_j)\theta_j^k$. Summing over $j$ proves \eqref{eq:finite_exponential}.
\end{proof}

\begin{proof}[Proof of \Cref{prop:order_m_saturation}]
The first $m$ mixed-cumulant coordinates satisfy a square Vandermonde system with nodes $\vartheta_1,\ldots,\vartheta_m$. Because the nodes are distinct, the Vandermonde matrix is nonsingular, so the weights $a_1^{(m)},\ldots,a_m^{(m)}$ exist and are unique. The remaining symmetric coordinate is the pure-$Y$ cumulant; choosing $b_m$ as stated matches that coordinate and completes the tensor representation.
\end{proof}

\begin{proof}[Proof of \Cref{cor:common_slope_dictionary}]
For each $\ell=2,\ldots,m$, the $\ell\times m$ Vandermonde matrix with entries $\vartheta_j^k$, $k=0,\ldots,\ell-1$, has full row rank because the $\vartheta_j$ are distinct. Hence weights solving the $\ell$ mixed-cumulant equations exist, although they need not be unique when $\ell<m$. Choosing $b_\ell$ to match the pure-$Y$ cumulant completes the tensor representation.
\end{proof}

\begin{proof}[Proof of \Cref{thm:prony-identification}]
Let $V$ be the $r\times r$ Vandermonde matrix with entries $V_{ij}=\theta_j^{i-1}$. Then
\begin{equation*}
\mathcal H_r(z)=V\diag(a_1,\ldots,a_r)V^\trans.
\end{equation*}
The distinctness of the $\theta_j$ and nonzero weights imply that $V$ and $\diag(a_1,\ldots,a_r)$ are nonsingular, hence $\mathcal H_r(z)$ is nonsingular.

Let $c$ solve \eqref{eq:annihilator-system}, and define $p$ by \eqref{eq:annihilator-poly}. Substituting $z_k=\sum_ja_j\theta_j^k$ into the $r$ recurrence equations gives
\[
\sum_{j=1}^ra_j\theta_j^i p(\theta_j)=0,
\qquad i=0,\ldots,r-1.
\]
In matrix form,
\[
V\diag(a_1,\ldots,a_r)
\begin{pmatrix}p(\theta_1)\\\vdots\\p(\theta_r)\end{pmatrix}=0.
\]
The first two factors are nonsingular, so $p(\theta_j)=0$ for every $j$. Since $p$ has degree $r$ and the $\theta_j$ are distinct, these are exactly its roots. The first $r$ equations in \eqref{eq:finite_exponential} then form a nonsingular Vandermonde system for the weights.
\end{proof}

\begin{proof}[Proof of \Cref{cor:capacity}]
If $r\le\lfloor m/2\rfloor$, then $m\ge2r$, so the coordinates $z_0^{(m)},\ldots,z_{2r-1}^{(m)}$ required by \Cref{thm:prony-identification} are available. The omitted-variable statement follows by substituting $r=q+1$.
\end{proof}

\begin{proof}[Proof of \Cref{cor:fourth-order}]
For $r=2$,
\[
H(z)=
\begin{pmatrix}1&1\\\theta_1&\theta_2\end{pmatrix}
\begin{pmatrix}a&0\\0&b\end{pmatrix}
\begin{pmatrix}1&\theta_1\\1&\theta_2\end{pmatrix}.
\]
Taking determinants gives \eqref{eq:det-h}. The annihilating polynomial is
\[
(t-\theta_1)(t-\theta_2)=t^2-(\theta_1+\theta_2)t+\theta_1\theta_2.
\]
Comparing with $t^2-c_1t-c_0$ gives \eqref{eq:symmetric-roots} and \eqref{eq:theta-pm}.
\end{proof}

\begin{proof}[Proof of \Cref{prop:ols-barycenter}]
Pairwise uncorrelatedness gives
\[
\Var(X)=\sum_{j=1}^r\sigma_j^2,
\qquad
\Cov(X,Y)=\sum_{j=1}^r\theta_j\sigma_j^2.
\]
Division gives \eqref{eq:ols-general-barycenter}. The weights are nonnegative and sum to one, implying \eqref{eq:convex-hull}.
\end{proof}

\begin{proof}[Proof of \Cref{prop:ovb-set}]
By \Cref{cor:fourth-order}, the fourth-order cumulants identify the unordered pair $\{\beta,\beta+\gamma\}$. Exchanging $H$ and $\nu$ leaves the observed equations unchanged after replacing $\beta$ by $\beta+\gamma$ and $\gamma$ by $-\gamma$. Hence the labels cannot be distinguished from the distribution of $(X,Y)$ alone. The remaining conclusions are immediate.
\end{proof}

\section{Proofs for fourth-order inference}

\begin{proof}[Derivation of \eqref{eq:z-clt}]
Let $g_4(X,Y)$ be the finite vector of monomials $X^aY^b$ with $a+b\le4$. Under \eqref{eq:eighth-moment}, every coordinate of $g_4$ has finite second moment. Hence the multivariate central limit theorem gives
\[
\sqrt n(\widehat M_4-M_4)\overset{d}{\longrightarrow}N(0,\Sigma_M).
\]
Each fourth-order joint cumulant is a polynomial in raw moments of total degree at most four, so $\phi$ is continuously differentiable. The multivariate delta method yields \eqref{eq:z-clt}--\eqref{eq:omega}. Consistency of the plug-in covariance follows from the law of large numbers applied to moments through degree eight.
\end{proof}

\begin{proof}[Derivation of \eqref{eq:jac-c}]
Write $c=H^{-1}q$. Differentiation gives
\[
\dd c=H^{-1}(\dd q-\dd H\,c).
\]
For perturbations of $(z_0,z_1,z_2,z_3)$ in the coordinate directions, the four vectors inside parentheses are respectively
\[
\binom{-c_0}{0},\quad
\binom{-c_1}{-c_0},\quad
\binom{1}{-c_1},\quad
\binom{0}{1}.
\]
Stacking them proves \eqref{eq:jac-c}.
\end{proof}

\begin{proof}[Proof of \Cref{thm:theta-clt}]
Under the stated conditions, the maps $z\mapsto c=H(z)^{-1}q(z)$ and $c\mapsto(\theta_-,\theta_+)$ are continuously differentiable near the truth. Their derivatives are \eqref{eq:jac-c} and \eqref{eq:jac-root}. Apply the delta method to \eqref{eq:z-clt} and compose the Jacobians. Plug-in consistency follows by continuity.
\end{proof}

\begin{proof}[Proof of \Cref{prop:h-test}]
The function $h(z)=z_0z_2-z_1^2$ is continuously differentiable with gradient in \eqref{eq:tau-h}. Under the one-slope null, $h(z)=0$. The result follows from the delta method applied to \eqref{eq:z-clt}.
\end{proof}

\begin{proof}[Derivation of \eqref{eq:tau-gaussian}]
Under $Y=\beta X+\varepsilon$ with $\varepsilon$ independent of $X$ and both symmetric, the plug-in cumulants are polynomials in even moments. Differentiating $h$ at the rank-one sequence $z_k=\kappa_4(X)\beta^k$ and using independence to drop odd cross terms leaves the influence function
\[
\kappa_4(X)\,(X^2-\E X^2)\,(\varepsilon^2-\E\varepsilon^2).
\]
The two factors are independent, so
\[
\tau_h^2=\kappa_4(X)^2\Var(X^2)\Var(\varepsilon^2).
\]
If $\Var(X)=1$, then $\Var(X^2)=\E[X^4]-1$. If $\varepsilon\sim N(0,\sigma_\varepsilon^2)$, then $\Var(\varepsilon^2)=2\sigma_\varepsilon^4$.
\end{proof}

\begin{proof}[Proof of \Cref{prop:local-n14}]
Under the two-slope model,
\[
h(z_n)=a_nb_n\delta_n^2=a_nb_nt^2n^{-1/2},
\]
so $\sqrt n\,h(z_n)\to abt^2$. By the assumed uniform cumulant CLT and delta expansion around $z_n$,
\[
\sqrt n\{\widehat h-h(z_n)\}\overset{d}{\longrightarrow}N(0,\tau_0^2).
\]
Adding the deterministic local shift proves \eqref{eq:local-h}.
\end{proof}

\begin{proof}[Proof of \Cref{prop:multiorder-gmm}]
Let $J=D_cR(c_0)$ and let $\Sigma_R$ denote the asymptotic covariance matrix of the stacked sample recurrence residuals at $c_0$. The assumed joint root-$n$ CLT for the relevant sample cumulants implies
\[
\sqrt n\,\widehat R(c_0)\overset{d}{\longrightarrow}N(0,\Sigma_R).
\]
Each recurrence residual is affine in $c$, with coefficients given by sample cumulants. Hence the sample Jacobian $\widehat J=D_c\widehat R(c)$ does not depend on $c$ and satisfies $\widehat J\overset{p}{\to}J$. The first-order condition for the quadratic minimum-distance criterion is therefore
\[
0
=
\widehat J^\trans W_n\{\widehat R(c_0)+\widehat J(\widehat c-c_0)\}.
\]
Because $J$ has full column rank and $W$ is positive definite, $J^\trans WJ$ is nonsingular, and the corresponding sample matrix is nonsingular with probability tending to one. Thus
\[
\sqrt n(\widehat c-c_0)
=
-(\widehat J^\trans W_n\widehat J)^{-1}
\widehat J^\trans W_n\sqrt n\,\widehat R(c_0),
\]
which yields root-$n$ consistency and, by Slutsky's theorem,
\[
\sqrt n(\widehat c-c_0)
\overset{d}{\longrightarrow}
N(0,V_c),
\qquad
V_c
=
(J^\trans WJ)^{-1}J^\trans W\Sigma_R WJ(J^\trans WJ)^{-1}.
\]
If $\Sigma_R$ is nonsingular, the efficient choice $W=\Sigma_R^{-1}$ reduces this to
\[
V_c=(J^\trans\Sigma_R^{-1}J)^{-1},
\]
the usual efficient minimum-distance covariance within this class. Finally, if the roots of $p_{c_0}$ are simple, each root is a continuously differentiable local function of $c$ by the implicit function theorem. Applying the delta method to these root maps gives root-$n$ asymptotic normality of the estimated latent slopes.
\end{proof}

\section{Proofs for the multivariate results}

\begin{proof}[Proof of \Cref{thm:multiv_recurrence}]
Let $K_U=\cum_m(U)$ and $K_V=\cum_m(V)$ be the symmetric order-$m$ cumulant tensors of the two vector components. By multilinearity and finite-order independence,
\[
Z_k=K_U\times\eta^{\otimes k}+K_V\times\beta^{\otimes k},
\]
where $k$ modes are contracted and $m-k$ remain free. Contracting $Z_{k+1}$ with $s=\beta+\eta$ yields a sum of pure terms and cross terms. Because $K_U$ and $K_V$ are symmetric,
\[
Z_k:Q
=K_U\times\eta^{\otimes k}\times\beta\times\eta
+K_V\times\beta^{\otimes k}\times\beta\times\eta.
\]
Subtracting cancels the cross terms and leaves
\[
K_U\times\eta^{\otimes(k+2)}+K_V\times\beta^{\otimes(k+2)}=Z_{k+2},
\]
proving \eqref{eq:multiv_recurrence}.
\end{proof}

\begin{proof}[Proof of \Cref{prop:multiv_testing_restrictions}]
From \eqref{eq:multiv_Z_representation} with $m=4$,
\[
Z_{k+1}-Z_k\times b
=
K_U\times\eta^{\otimes k}\times(\eta-b)
+K_V\times\beta^{\otimes k}\times(\beta-b).
\]
Using \eqref{eq:multiv_centered_slopes} with $k=0$ gives \eqref{eq:multiv_C_signal}. Applying the same identity once more gives
\[
M_2(b)
=K_U[(\eta-b),(\eta-b)] + K_V[(\beta-b),(\beta-b)],
\]
and
\[
M_3(b)
=K_U[\eta,(\eta-b),(\eta-b)] + K_V[\beta,(\beta-b),(\beta-b)].
\]
Substituting \eqref{eq:multiv_centered_slopes} yields \eqref{eq:multiv_M2_signal}--\eqref{eq:multiv_M3_signal}.
\end{proof}

\begin{proof}[Proof of \Cref{thm:multiv_wald_test}]
Let $G_4(X,Y)$ collect all monomials in $(X^\top,Y)^\top$ of total degree at most four.  The eighth-moment assumption gives a multivariate CLT for the empirical mean of $G_4$.  The map from these raw moments to $\widehat b$, the empirical cumulant tensors $(\widehat Z_0,\ldots,\widehat Z_3)$, and hence $\widehat\psi_Q$, is continuously differentiable in a neighborhood of the null whenever $\Sigma_X$ is nonsingular.  The multivariate delta method therefore gives \eqref{eq:multiv_psi_clt}, and plug-in estimation yields a consistent $\widehat\Omega_Q$.

Under locally constant rank, the standardized projection of $\sqrt n\widehat\psi_Q$ onto the column space of $\Omega_Q$ converges to an $r_Q$-dimensional standard normal vector.  The quadratic form in \eqref{eq:multiv_Q_stat} therefore converges to $\chi^2_{r_Q}$.  Under a fixed alternative $P$, $\widehat\psi_Q\xrightarrow{p}\psi_Q(P)$ and $\widehat\Omega_Q^\dagger\xrightarrow{p}\Omega_Q(P)^\dagger$ under the stated stable-rank condition. Hence
\[
\frac{T_{Q,n}}{n}
\xrightarrow{p}
\psi_Q(P)^\top\Omega_Q(P)^\dagger\psi_Q(P).
\]
Condition \eqref{eq:multiv_visible_alt} therefore implies divergence of $T_{Q,n}$ and proves consistency against such alternatives.
\end{proof}


\begin{proof}[Proof of \Cref{prop:multiv_local_power}]
Write $\delta_n=h n^{-1/4}$ and use \eqref{eq:multiv_M2_signal}.  Since $(I-A_n)\delta_n=n^{-1/4}u+o(n^{-1/4})$ and $A_n\delta_n=n^{-1/4}v+o(n^{-1/4})$,
\[
\sqrt n\,M_{2,n}\to K_U[u,u]+K_V[v,v].
\]
For $M_{3,n}$, replace $\eta_n$ and $\beta_n$ by their common limit $\theta$ in the leading quadratic term.  The additional term generated by $\eta_n-\theta=O(n^{-1/4})$ is cubic in $\delta_n$ and is $o(n^{-1/2})$, which gives \eqref{eq:multiv_local_M3}.  Combining these deterministic shifts with the assumed uniform CLT and delta expansion gives \eqref{eq:multiv_local_normal}.  The noncentral chi-square limit follows by the same projection argument as in \Cref{thm:multiv_wald_test}.
\end{proof}

\begin{proof}[Proof of \Cref{cor:multiv_ovb}]
Because $\eta=\beta+\Delta\gamma$ and $\Delta^\top\Delta=I_q$, we have $\eta=\beta$ if and only if $\gamma=0$.  Moreover,
\[
\Cov(X,Y)=\Sigma_X\beta+\Delta\Sigma_H\gamma,
\]
so multiplication by $\Sigma_X^{-1}$ gives
$$
\beta_{\LS}-\beta=\Sigma_X^{-1}\Delta\Sigma_H\gamma.
$$ 
If $\beta_{\LS}=\beta$, then $\Delta\Sigma_H\gamma=0$; premultiplying by $\Delta^\top$ and using $\Sigma_H\succ0$ gives $\gamma=0$.  The converse is immediate.
\end{proof}

\begin{proof}[Proof of \Cref{thm:multiv_identification}]
Injectivity of $\mathcal L_Z$ implies that \eqref{eq:multiv_rec_1}--\eqref{eq:multiv_rec_2} have at most one solution $(s,Q)$; existence follows from \Cref{thm:multiv_recurrence}. Thus $(s,Q)$ is identified. Direct expansion gives
\[
ss^\top-4Q
=(\beta+\eta)(\beta+\eta)^\top
-2(\beta\eta^\top+\eta\beta^\top)
=(\beta-\eta)(\beta-\eta)^\top.
\]
If $\beta\neq\eta$, the matrix is rank-one positive semidefinite. Any factor $d$ satisfying \eqref{eq:multiv_d_factor} equals either $\beta-\eta$ or $\eta-\beta$, which gives \eqref{eq:multiv_beta_eta_recovery}.
\end{proof}

\end{document}